\documentclass[aps,prl,twocolumn,superscriptaddress,longbibliography]{revtex4-2}

\usepackage{graphicx}
\usepackage{bm}
\usepackage{microtype} 
\usepackage{amsmath}
\usepackage{amssymb}
\usepackage[table]{xcolor}
\usepackage{array}
\usepackage{multirow}
\usepackage[caption=false]{subfig}
\usepackage{physics}
\usepackage[normalem]{ulem} 
\usepackage[verbose]{placeins}
\usepackage{amsthm}
\usepackage{bbm}
\usepackage[colorinlistoftodos]{todonotes}
\usepackage[scr=boondoxo,frak=boondox,scrscaled=1.05]{mathalfa}

\def\relu{\mathrm{ReLU}}

\def\bra#1{\langle#1|}
\def\ket#1{|#1\rangle}
\def\braket#1#2{\langle#1|#2\rangle}

\renewcommand\vec{\boldsymbol}

\usepackage{hyperref}
\hypersetup{
    breaklinks=true,    
    colorlinks=true,    
    linkcolor=violet,   
    citecolor=violet,  
    urlcolor=violet
}

\begin{document}
\raggedbottom
\title{Bound on entanglement in neural quantum states}

\author{Nisarga Paul}
\email{npaul@caltech.edu}
\affiliation{Department of Physics, Massachusetts Institute of Technology, Cambridge, Massachusetts 02139, USA}
\affiliation{Department of Physics and Institute for Quantum Information and Matter,
California Institute of Technology, Pasadena, California 91125, USA}

\begin{abstract}
Variational wavefunctions offer a practical route around the exponential complexity of many-body Hilbert spaces, but their expressive power is often sharply constrained. Matrix product states, for instance, are efficient but limited to area law entangled states. Neural quantum states (NQS) are widely believed to overcome such limitations, yet little is known about their fundamental constraints. Here we prove that feed-forward neural quantum states acting on $n$ spins with $k$ scalar nonlinearities, under certain analyticity assumptions, obey a bound on entanglement entropy for any subregion: $S \leq c k\log n$, for a constant $c$. This establishes an NQS analog of the area law constraint for matrix product states and rules out volume law entanglement for NQS with $O(1)$ nonlinearities. We demonstrate analytically and numerically that the scaling with $n$ is tight for a wide variety of NQS. Our work establishes a fundamental constraint on NQS that applies broadly across different network designs, while reinforcing their substantial expressive power.
\end{abstract}

\maketitle

\paragraph*{Introduction.} A significant challenge in quantum many-body physics is representing and manipulating states in the exponentially large Hilbert space. A common strategy is to restrict to variational states tailored to the needs of the problem. Any variational state must strike a balance between efficiency and expressivity: at one extreme, general wavefunctions quickly become computationally intractable, while at the other, few-parameter wavefunctions may be insufficiently expressive. \par 

In tensor network-based states such as matrix product states (MPS)~\cite{white1992density,schollwock2011density,orus2019tensor,verstraete2008matrix,cirac2021matrix}, this tradeoff is well understood. While they are amenable to very efficient algorithms, their expressive power is constrained by an area law upper bound on entanglement. MPS faithfully describe ground states of gapped one-dimensional systems with short-range interactions, but cannot efficiently represent critical states or states with volume law entanglement. Knowing these limitations has been crucial: it underpins our understanding of their success for gapped ground states~\cite{hastings2007area,arad2013area,landau2015polynomial} and points to where new methods are needed~\cite{verstraete2004renormalization,PhysRevLett.99.220405}. 
\par 
Neural quantum states (NQS), introduced by Carleo and Troyer~\cite{carleo2017solving}, have emerged as an alternative variational class based on neural networks. They capitalize on the remarkable expressivity of neural networks and efficiency of modern optimizers for a huge number of parameters~\cite{sharir2022neural,levine2019quantum,raghu2017expressive}. In practice, NQS have proven competitive across condensed matter, quantum chemistry, and nuclear physics~\cite{luo2019backflow,cassella2023discovering,chen2024empowering,rende2024simple,teng2025solving,geier2025self,pfau2020ab,hermann2020deep,pfau2024accurate,yang2023deep,adams2021variational,gnech2024distilling,fore2025investigating}, often achieving state-of-the-art accuracy and competitive scaling compared to tensor network and quantum Monte Carlo~\cite{giuliani2023learning,PhysRevLett.134.079701,kufel2025approximately,gao2017efficient,deng2017quantum}. Yet, despite this empirical progress, our theoretical understanding of NQS remains nascent. In contrast to MPS, where our understanding of expressivity constraints is mature, clear constraints are largely lacking for NQS. Identifying such constraints, if they exist, is crucial for conceptual understanding and for deploying targeted solutions.
\par 

Due to their variability and flexibility, establishing general constraints for NQS is a difficult task. Moreover, by the universal approximation theorem~\cite{hornik1989multilayer}, sufficiently large NQS can in principle approximate any state. Therefore, it is important to refine the scope of the problem. In this Letter, we choose to restrict the number of scalar nonlinearities $k$ in the computational graph, \textit{i.e.} the number of elementary nonlinear operations or neurons. Specifically, we consider families of NQS defined for arbitrary system sizes $n$, but with $k$ growing strictly slower than $O(n/\log n)$; for clarity, we often state results in the regime $k= O(1)$. Restricting to these NQS has two advantages. First, these states are computationally cheaper, typically reducing the cost of variational Monte Carlo (VMC). Second, as we'll show, bounding $k$ allows us to derive a sharp constraint on entanglement.  

We prove that for these NQS, entanglement entropy of a subregion $A$ satisfies $S_A = O(\mu\log n)$, where $\mu\leq k+1$ is the effective dimension of an intermediate feature space. Thus, no family of NQS can achieve volume law entanglement, which requires $S_A = O(|A|)$ for $|A|= O(n)$, unless $k$ is at least $O(n/\log n)$, subject to our assumptions. This provides a neural quantum state analog of the area law constraint in MPS. For the case $k=O(1)$, we demonstrate analytically and numerically that the scaling in $n$ is tight. The bound is quite agnostic to architectural choices, holding across a wide range of NQS. 
\par 
The key idea of our bound is that the amplitudes of a feed-forward NQS with $k$ nonlinearities only depends on at most $k+1$ ``collective variables", or affine combinations of the underlying degrees of freedom, rather than on all $n$ of them. Under a further analyticity assumption, which we argue is mild, this dependence can be approximated by sufficiently low-degree polynomials, which leads to the bound. In summary, the number of elementary nonlinear operations $k$ in the NQS architecture acts as a bottleneck for its entanglement. \par 

Beyond the main theorem, our method can be applied to a broad class of NQS to derive entanglement bounds, and should find further use. For example, in multilayer perceptrons with polynomial nonlinearities, we show that achieving volume law entanglement requires either large width or depth. \par 
We emphasize that though our main result is an upper bound, the bound is also saturated quite generally: for a variety of cases with $k\sim O(1)$ we find $\log n$ entanglement scaling. This surpasses the area law entanglement capacity of MPS in 1D, underscoring the remarkable expressive power of NQS even with very few nonlinearities. 
\par 


\par 

\paragraph*{Setup.} We consider a pure quantum state of $n$ discrete degrees of freedom $\vec s = (s_1,\ldots, s_n)$, which we refer to as spins but could also be fermion occupation numbers. The state is fully specified by the wavefunction amplitudes $\braket{\vec s}{\Psi} =  \Psi(\vec s)$. We assume $\Psi(\vec s)$ is given by a feed-forward neural network~\footnote{We assume a feed-forward neural network to mean the computational graph is directed and acyclic. Most architecture (e.g. multilayer perceptrons, convolutional neural networks, transformers) are included in this definition. Recurrent neural networks are not included in general; however, when evaluated for a finite horizon $T$ they can be unrolled into a DAG. Our entanglement bound then applies to the unrolled DAG with $k$ replaced by the effective number of scalar nonlinearities $k_{\rm eff}$.} and first establish some structural properties of $\Psi(\vec s)$.\par

The computation of $\Psi(\vec s)$ can generally~\cite{sharir2022neural} be described by a directed acyclic graph $(V,E)$ which maps inputs $s_1,\ldots, s_n$ to the output $\Psi(\vec s)$. By acyclic we mean there are no directed cycles in the graph. We illustrate an example of such a graph in Fig.~\ref{fig:dag_nqs}. The values of the nodes $v\in V$ are defined recursively by
\begin{equation}
    v(\vec s) = \sigma_v\left(b_v + \sum_{(u,v)\in E}w_{u,v}u(\vec s) + \sum_i w_{i,v} s_i\right)
\end{equation}
where $\sigma_v$ is a scalar nonlinearity, the first sum is over directed edges $(u,v)\in E$ pointing to $v$, and the second sum is over input spins, thus allowing skip connections. The weights and biases $b_v, w_{u,v}, w_{i,v}$ are network parameters for optimization. We'll assume the network parameters are real but allow the nonlinearities and network output to be complex. 
\par

\begin{figure}
\begin{tikzpicture}[x=1.3cm, y=0.5cm, >=stealth]
  \tikzstyle{node}=[circle,draw=black,minimum size=9pt,inner sep=0pt]
  \tikzstyle{input}=[node,fill=white!15]
  \tikzstyle{hidden}=[node,fill=green!15]
  \tikzstyle{output}=[node,fill=red!15]

  \node[input,label={$\vec s$}] (I1) at (0,-0.5) {};
  \node[input] (I2) at (0,-1.5) {};
  \node[input] (I3) at (0,-2.5) {};
  \node[input] (I4) at (0,-3.5) {};
  
  \node[hidden] (H1) at (1.5,-1.0) {};
  \node[hidden] (H2) at (1.5,-2.7) {};
  \node[hidden] (H3) at (3.0,-1.85) {};
  \node[hidden] (H4) at (2.2,-3.0) {};
  \node[hidden] (H5) at (2.5,-0.75) {};

  \node at (2.0,0.2) {$v \in V$};

  \node[output,label={$\Psi(\vec s)$}] (O) at (4.0,-1.8) {};

  \draw[->,thin] (I1) -- (H1);
  \draw[->,thin] (I2) -- (H1);
  \draw[->,thin] (I3) -- (H1);
  \draw[->,thin] (I4) -- (H1);

  \draw[->,thin] (I1) -- (H2);
  \draw[->,thin] (I2) -- (H2);
  \draw[->,thin] (I4) -- (H2);

  \draw[->,thin] (H1) -- (H3);
  \draw[->,thin] (H2) -- (H3);
  \draw[->,thin] (H2) -- (H4);

  \draw[->,thin] (H2) -- (H5);
  \draw[->,thin] (H4) -- (H5);
  \draw[->,thin] (H1) -- (H5);

  \draw[->,thin] (H3) -- (O);
  \draw[->,thin] (H4) -- (O);
  \draw[->,thin] (H5) -- (O);

\end{tikzpicture}
\caption{Example of a directed acyclic graph representation of a feed-forward neural quantum state. Input spins $\vec s=(s_1,\ldots, s_n)$ (white) are mapped through internal nodes $v\in V$ (green) to an output $\Psi(\vec s)$ (red). Internal nodes represent affine combinations followed by scalar nonlinearities, while the output node is an affine combination. We show how the number of nonlinearities $k=|V|$ acts as a bottleneck for entanglement.}
\label{fig:dag_nqs}
\end{figure}

\par 
Let us define an \textit{affine feature} of $\vec s$ to be a scalar of the form $t(\vec s) = \vec u^T \vec s + c$. It is instructive to interpret $\Psi$ as a function of some number of affine features of $\vec s$; that is, to find a function $\mathcal{G}$ such that 
\begin{equation}\label{eq:FequalG}
    \Psi(\vec s) = \mathcal{G}(t_1(\vec s),\ldots, t_\mu(\vec s)).
\end{equation}
While it is always possible to take $t_j(\vec s) = s_j$ and $\mathcal{G} = \Psi$, the representation Eq.~\eqref{eq:FequalG} will be useful only when $\mu$ is less than $n$. As an example of when this holds, consider if $\Psi$ were purely linear. Then $\Psi(\vec s)$ must be of the form $\vec w^T\vec s + b$, and $\mathcal{G}$ exists with $\mu=1$. \par 

More generally, when the network $\Psi$ consists of $k$ nonlinearities, $\mathcal{G}$ can be found with $\mu \leq k+1$ affine features. We call this \textit{feature reduction} and provide a full constructive proof in the Supplemental Material~\cite{supp}. The basic reasoning is as follows. A directed acyclic graph (DAG) has a node ordering such that if a directed edge $(u,v)$ connects $u$ and $v$, $u$ precedes $v$ in the ordering. We may proceed inductively on such an ordering of the DAG of $\Psi$. The first node must be a function of a single affine feature, and each subsequent node is a function of the preceding affine features and possibly one more. Including an optional overall affine feature then gives $\mu \leq k+1$. \par 

Intuitively, this means that a feed-forward network with $k$ nonlinearities
cannot couple the $n$ input spins through more than $k+1$ independent
affine features (or collective variables). Thus if $k$ is small, the wavefunction amplitudes are functions of a
small number of affine features rather than all $n$ degrees of freedom. This will be the key to bounding entanglement.\par

Next, we must make some well-behavedness assumptions about the NQS. Different choices can be adopted here, but the goal is to make a reasonable assumption prohibiting the NQS from ``hiding" a large amount of complexity in a few nonlinearities. We will choose to require the following technical condition: that $\mathcal{G}(t_1,\ldots, t_\mu)$ is analytic in a complex neighborhood of each of its real arguments, and bounded on this neighborhood, in a manner independent of $n$. In particular, we assume the domains of the $t_i$ are bounded independently of $n$, which sets the scaling of weights to be $\|\vec w\|_1 = O(1)$. We address relaxations of this assumption in the Supplemental Material~\cite{supp}. With this assumption, we may use standard results in approximation theory~\cite{trefethen,supp} that there exists a multivariable polynomial of degree at most $d$ in each variable that approximates $\mathcal{G}$ uniformly on its domain $\mathcal{D}\subset \mathbb{R}^\mu$ to accuracy $O(e^{-\gamma d})$, where $\gamma>0$ is a constant. 
\par 
In fact, our condition can be phrased as a quite natural condition on the nonlinearities: that they are analytic and bounded on a neighborhood of their domains, in a manner independent of $n$. This prohibits any single nonlinearity from hiding a complexity which grows with $n$. Hence our condition encompasses many standard activations used in practice, including
$\tanh$, softplus, and polynomials~\cite{nair2010rectified}.

\paragraph*{Entanglement bound.} At this stage, we have established that a feed-forward NQS with $k$ nonlinearities can be written in the form Eq.~\eqref{eq:FequalG} with $\mu \leq k+1$ and, with a suitable analyticity assumption on $\mathcal{G}$, for each $d$ there exists a polynomial $P_d$ of degree no more than $d$ in each variable such that
\begin{equation}\label{eq:uniform}
    \max_{\vec x \in \mathcal{D}} |\mathcal{G}(\vec x) - P_d(\vec x)| \leq \beta e^{-\gamma d}.
\end{equation}
Here, $\vec x= (x_1,\ldots, x_\mu)$, $\mathcal{D}$ is the domain of $\mathcal{G}$, and $\beta,\gamma$ are $O(1)$ constants. We are now in a position to state our main result. \par 
\textbf{Theorem. }\textit{(Entanglement bound)} Let $\Psi(\vec s)$ be a feed-forward NQS on $n$ spins with $k$ scalar nonlinearities. We have that $\Psi(\vec s) = \mathcal{G}(t_1(\vec s),\ldots, t_\mu(\vec s))$ where $t_i(\vec s)$ are affine features and $\mu \leq k+1$. Suppose $\mathcal{G}$ extends analytically and remains bounded in a complex neighborhood of each of its real arguments, in a manner independent of $n$. Let $A$ be a subset of the $n$ spins 
and let $S_A$ be the von Neumann entropy of the reduced density matrix 
$\rho_A = \Tr_{\bar A} \ket\Psi\bra\Psi$. Then the entropy grows at most linearly 
with $\mu$ and at most logarithmically in $n$:
\begin{equation}\label{eq:SAklogn}
    S_A = O(\mu \log n).
\end{equation}
In particular, no feed-forward NQS with $k=O(1)$ nonlinearities can exhibit volume law entanglement. We provide the formal theorem and proof in the Supplemental Material~\cite{supp}. In the following we outline the proof strategy. 
\par 

To proceed, we first introduce the auxiliary state 
\begin{equation}
    \ket{\tilde \Psi} = \mathcal{N}^{-1/2} \sum_{\vec s} P_d(t_1(\vec s),\ldots, t_\mu(\vec s)) \ket{\vec s}
\end{equation}
where $P_d$ satisfies Eq.~\eqref{eq:uniform} and $\mathcal{N}$ ensures unit normalization. Letting $\tilde \rho_A = \Tr_{\bar A} \ket{\tilde \Psi}\bra{\tilde \Psi}$ and $\tilde S_A = -\Tr \tilde \rho_A \log \tilde \rho_A$, we will first seek a bound on $\tilde S_A$.\par 

Let $\mathcal{H}_{A}, \mathcal{H}_{\bar A}$ denote the Hilbert spaces of $A$ and its complement $\bar A$, respectively, and let $\ket{\vec u}_A,\ket{\vec v}_{\bar A}$ denote corresponding basis states. We define a $\dim \mathcal{H}_A \times \dim \mathcal{H}_{\bar A}$ matrix $\widetilde M$ by 
\begin{equation}\label{eq:auxiliary}
    \ket{\tilde \Psi} := \mathcal{N}^{-1/2}\sum_{\vec u,\vec v} \widetilde M_{\vec u\vec v} \ket{\vec u}_A \ket{\vec v}_{\bar A}.
\end{equation}
We have that $\tilde S_A \leq \log \tilde r$, where $\tilde r = \rank \widetilde M$ is the Schmidt rank of $\ket{\tilde \Psi}$. Let us write
\begin{equation}\label{eq:Pd}
    P_d(\vec x) = \sum_{0\leq n_1,\ldots, n_\mu \leq d} \alpha_{\vec n} x_1^{n_1} \cdots x_\mu^{n_\mu}
\end{equation}
where $\alpha_{\vec n}$ are scalar coefficients. Furthermore, let us introduce \textit{split affine features} $x_i(\vec u)$ and $y_i(\vec v)$, which are affine features of $\vec u$ and $\vec v$ respectively such that $x_i(\vec u) + y_i(\vec v) = t_i(\vec s)$ for each $i$. An explicit expression for $\widetilde M$ using Eq.~\eqref{eq:Pd} and binomial expansion is then
\begin{multline}
    \widetilde M =  \sum_{0\leq n_1,\ldots, n_\mu \leq d}\sum_{k_1,\ldots, k_\mu=0}^{n_1,\ldots,n_\mu} \alpha_{\vec n}\times\\
    \left(\sum_{\vec u}\prod_{i=1}^\mu \binom{n_i}{k_i} [x_{i}(\vec u)]^{n_i-k_i} \ket{\vec u}\right)\left(\sum_{\vec v}\prod_{j=1}^\mu [y_{j}(\vec v)]^{k_{j}}\bra{\vec v}\right).
\end{multline}
Each summand is an outer product of rank one. Hence by subadditivity of rank, the rank of $\widetilde M$ is bounded by $\sum_{0\leq n_1,\ldots, n_\mu \leq d}\sum_{k_1,\ldots, k_\mu=0}^{n_1,\ldots,n_\mu} 1$. We conclude that
\begin{equation}\label{eq:rtilde}
    \tilde r \leq ((d+1)(d+2)/2)^\mu, 
\end{equation}
and, using $\tilde S_A \leq \log \tilde r$, we have $\tilde S_A = O(\mu \log d)$. In words, the auxiliary state's Schmidt decomposition produces only polynomially many terms, so the Schmidt rank is at most poly($d$) and the entropy at most logarithmic in $d$.\par 

Next, we convert this to a bound on $S_A$ using the approximation Eq.~\eqref{eq:uniform}. This requires a few steps. First, if we choose a sufficiently large degree $d= \Theta(n)$ for the polynomial $P_d$, we can ensure that $\|\ket\Psi - \ket{\tilde \Psi}\|_2 = O(e^{-\kappa n})$ for some $\kappa>0$. Next, with standard arguments we may bound the trace distance of reduced density matrices, $\|\rho_A - \tilde \rho_A\|_1 = O(e^{-\kappa n})$. Finally, we may apply the Fannes-Audenaert inequality 
\begin{equation}\label{eq:FAineq}
 |S_A - \tilde S_A| \leq T \log(2^{|A|}-1) +H_2(T)
\end{equation}
where $T = \frac12 \|\rho_A-\tilde \rho_A\|_1$ and $H_2(x) = -x\log x -(1-x)\log(1-x) < 1$ to conclude $|S_A - \tilde S_A| = O(1)$. Hence the entanglement bound on the auxiliary state transfers to $\ket\Psi$, and Eq.~\eqref{eq:SAklogn} follows due to the choice $d=\Theta(n)$. \par 

Theorem~\eqref{eq:SAklogn} as stated does not cover nonanalytic activations such as $\relu(x)=\max(0,x)$; nevertheless, our numerical experiments suggest the same scaling behavior, and smooth approximations to ReLU (e.g. GeLU) fall within our assumptions. We also note that a recurrent NQS evaluated for a finite horizon $T$ can be unrolled into a DAG with an effective number of scalar nonlinearities $k_{\rm eff}$. If the unrolled DAG satisfies our analyticity assumptions, then the bound applies with $k$ replaced by $k_{\rm eff}$~\cite{supp}.

\begin{figure}
    \centering
\includegraphics[width=\linewidth]{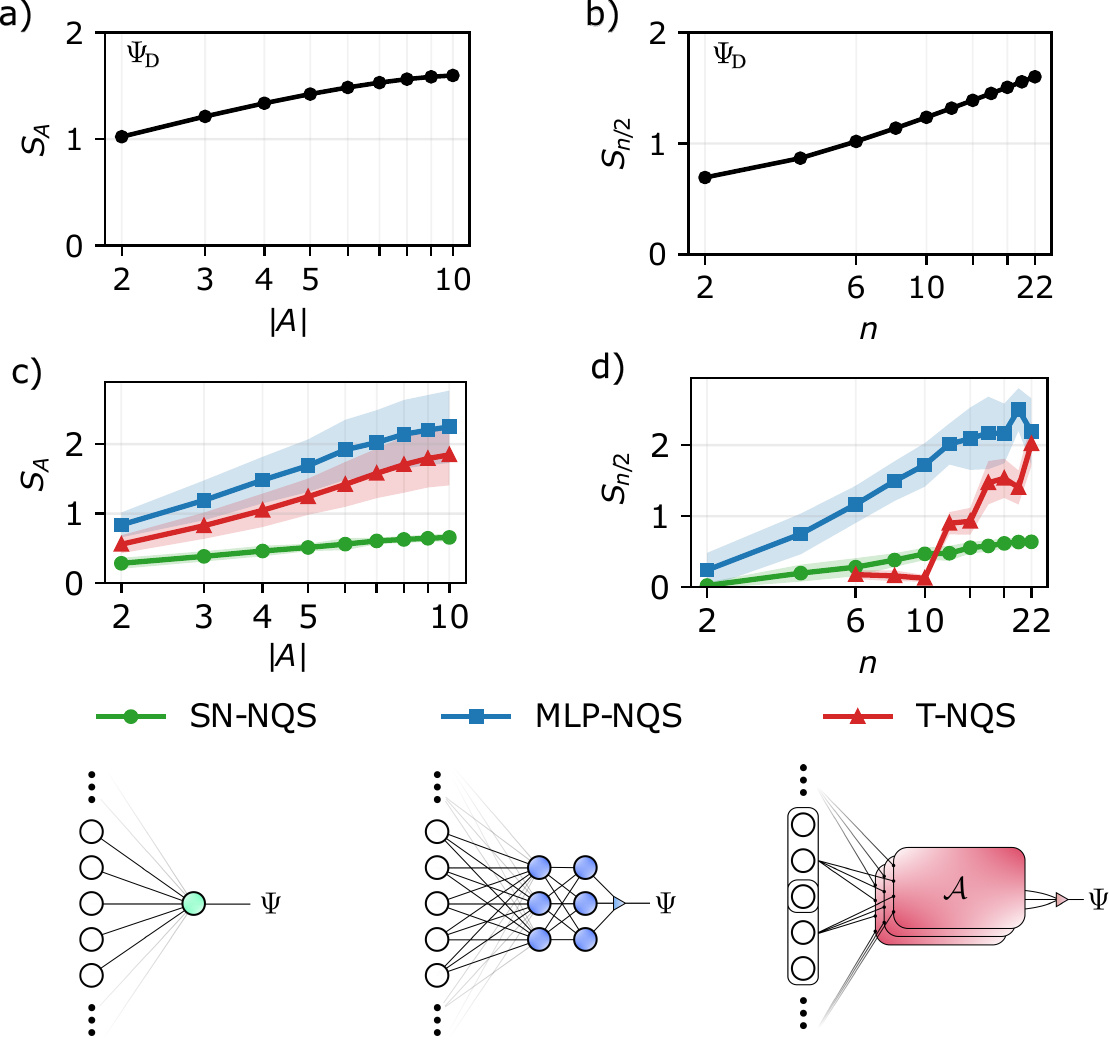}
    \caption{(a) Subregion von Neumann entanglement entropy $S_A$ for $n=22$ and (b) bipartite von Neumann entanglement entropy $S_{n/2}$ for the Dicke state, Eq.~\eqref{eq:Dicke}, on a logarithmic $x$-axis, corroborating analytically predicted scaling, respectively, up to finite-size effects. (c,d) Subregion and bipartite entanglement entropy for single-nonlinearity, multilayer perceptron, and transformer neural quantum states (SN-NQS, MLP-NQS, T-NQS) on a logarithmic $x$-axis with random parameters and shading indicating one standard deviation. The MLP-NQS has width $3$ and depth $2$. For the T-NQS, we chose patches of size $6$, stride $5$, 4 attention heads, and 2 layers. Each point was averaged over 20 trial wavefunctions.}
    \label{fig:numerics}
\end{figure}
Let us now establish, by explicit example, that the scaling of the bound with $n$ is tight. We choose to illustrate this using what we refer to as the Dicke state~\cite{dicke1954coherence,stockton2003characterizing}, an equal-weight superposition of states with vanishing total spin. The Dicke state can be expressed as a NQS with a single nonlinearity: 
\begin{equation}\label{eq:Dicke}
    \ket{\Psi_{\mathrm{D}}} := \sum_{\vec s }\delta_{\Sigma s_i,0} \ket{\vec s} =  \sum_{\vec s} \sigma\left(\vec w^T\vec s\right)\ket{\vec s}
\end{equation}
where $\vec w = (1,\ldots, 1)$ and $\sigma(x)$ is any nonlinearity which evaluates to $\delta_{x0}$ on integers $x$ (here we'll assume $s_i=\pm 1$). Let $A$ be an extensive subregion, $|A| \sim O(n)$. Then, for large $n$, $S_A = \frac12 \log n  + O(1)$~\cite{supp}. This shows the asymptotic scaling of our bound is optimal. Moreover, this shows logarithmic entanglement scaling can be achieved by NQS with even a single nonlinearity. We plot finite-size numerical results for this state in Fig.~\ref{fig:numerics}.\par 

\paragraph*{Numerical support.} Next, we show numerically that the scaling of our bound with $n$ appears to be saturated by a wide variety of feed-forward NQS, within accessible sizes. For this purpose, we choose three architectures. First, a single-nonlinearity neural quantum state (SN-NQS), whose log-amplitudes, up to normalization, are given by $\sigma(\vec w^T \vec s + b)$
where $\sigma = \sigma_1 + i\sigma_2$ and $\sigma_{1,2}$ are real nonlinearities. Second, a multilayer perceptron neural quantum state (MLP-NQS), which processes inputs $\vec s$ through a vanilla neural network of width $w_0$ and depth $d_0$ before linear projections to the real and imaginary parts of the log-amplitude. And third, a transformer neural quantum state (T-NQS), which splits $\vec s$ into $O(1)$ patches, processes them using a sparse self-attention mechanism~\cite{vaswani2017attention,viteritti2023transformer}, concatenates the resulting vectors, and linearly projects onto the real and imaginary parts of the log-amplitude~\footnote{Note the transformer architecture includes vector nonlinearities. These can be decomposed into combinations of scalar nonlinearities and linear operations. The transformer architecture with $O(1)$ patches, heads, and layers will thus still have $O(1)$ scalar nonlinearities and belong to our framework.}. These are shown in Fig.~\ref{fig:numerics}. \par

For each of these architectures, we calculate the entanglement entropy of a subset $A$ of $n$ qubits. In each case, the total number of nonlinearities $k$ remains $O(1)$. We treat free parameters as independent random variables drawn from normal distributions and choose $\tanh$ nonlinearities (see \cite{supp} for further details). In Fig.~\ref{fig:numerics}c, where we fix $n$ and vary $|A|$, we find that entanglement in each ensemble of states grows logarithmically. This saturates our bound, which constrains $S=O(\log n)$ for $O(1)$ nonlinearities. In Fig.~\ref{fig:numerics}d, we fix $|A|=n/2$ and vary $n$, finding scaling in $n$ clearly consistent with our bound for SN-NQS and MLP-NQS. \par 

Our main results indicate that volume law entanglement cannot be achieved by an NQS on $n$ spins with $O(1)$ nonlinearities. Does this align with known examples of volume law NQS? A simple example of a neural quantum state reaching volume law entanglement is the cosine network (CosNet)~\cite{luo2023infinite,halverson2021building}, a type of MLP-NQS, defined by $\Psi = \Psi_1 + i\Psi_2$ and each of $\Psi_{1,2}$ of the form
\begin{equation}
    f(\vec s) = \sum_{i=1}^k a_i \cos(\vec w_i^T \cdot \vec s + b_i).
\end{equation}
It has coefficients $a_i \sim \mathcal{N}(0,\sigma_a^2/k)$, weights $w_i \sim \mathcal{N}(0,\sigma_w^2/n)$, and biases $b_i \sim \mathcal{U}[-\pi,\pi]$. We plot the entanglement entropy for CosNet in Fig.~\ref{fig:cosnet}, showing that it is consistent with our expectations: the entanglement remains sub-volume-law until $k$ becomes comparable to $n$. In this case, the entanglement entropy at fixed $|A|$ also displays logarithmic growth as a function of $k$. This scaling in $k$ is again consistent with our upper bound, though it does not saturate it. 
\par 
\begin{figure}
    \centering
    \includegraphics[width=\linewidth]{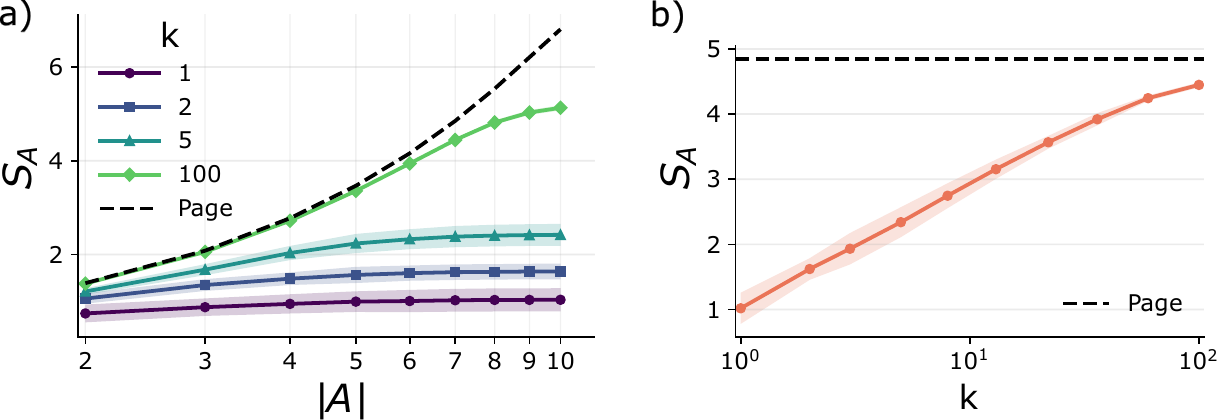}
    \caption{(a) Subregion entanglement entropy for CosNet states with $n=22$ qubits, averaged over 20 random initializations. The dashed black line is the Page value, \textit{i.e.} the Haar-random state prediction. Consistent with (sub)logarithmic entanglement scaling for $k\ll n$ and appears volume law for $k\gg n$. (b) $S_{|A|=7}$ for varying $k$. Plots have logarithmic $x$-axes, shading indicates one standard deviation, and $(\sigma_a,\sigma_w)=(10,1)$. }  \label{fig:cosnet}
\end{figure}
\par 
So far, we have not imposed a geometry on the $n$-spin system. Let us now consider spins on a $d$-dimensional lattice of linear size $L$, with $n=L^d$. In this case, we may apply our bound to area law entangled states, which satisfy $S = O(L^{d-1})$ for extensive subregions. In particular, area law entangled states require NQS with number of nonlinearities $k$ scaling at least as $O(L^{d-1}/\log L)$ under our assumptions. Therefore in dimensions $d>1$ even area law entangled states require a number of nonlinearities growing with system size for accurate NQS representation. This includes, for instance, gapped ground states of local Hamiltonians, which in one dimension are rigorously known and in higher dimensions are widely believed to obey an entanglement area law.

\par
Finally, the framework we've presented can be adapted to rigorously bound entanglement in a variety of specific settings. As a sharp example, consider a MLP-NQS of arbitrary width $w_0$ and depth $d_0$ with polynomial nonlinearities of degree at most $h$. Following the arguments of Eq.~\eqref{eq:auxiliary}-\eqref{eq:rtilde}, we can conclude that the entanglement entropy of any subregion is bounded by 
\begin{equation}
    S \leq w_0 \log\left[(h^{d_0}+1)(h^{d_0}+2)/2\right]\approx 2w_0 d_0\log h,
\end{equation}
where the approximation is valid for large $h^{d_0}$. This implies that either large width or large depth is strictly necessary for MLP-NQS with polynomial nonlinearities to achieve volume law entanglement.

\paragraph{Outlook. } In this work, we have provided an upper bound on entanglement in feed-forward neural quantum states with $k$ elementary scalar nonlinear operations acting on $n$ spins or lattice fermion degrees of freedom, namely $S = O(k\log n)$ as $n$ tends to infinity. Analytical and numerical examples confirm that this upper bound is saturated for many generic types of NQS. While our main contribution is an upper bound on entanglement, our work also affirms the substantial expressive power of NQS: logarithmic entanglement can generally be achieved even with very few nonlinearities. In contrast, logarithmic entanglement scaling in MPS requires a bond dimension growing polynomially with system size. 

\par 

Our theorem applies to feed-forward (directed acyclic) NQS subject to analyticity assumptions and in which the individual scalar nonlinearities are independent of $n$.
It therefore clarifies when volume-law entanglement is and is not compatible with NQS. In particular, volume-law entanglement is permitted by our theorem through several mechanisms:
\textit{(i)} scaling the number of scalar nonlinearities as $k=\Omega(n/\log n)$, as seen in many constructions where the number of neurons grows with $n$~\cite{luo2023infinite,deng2017quantum};
\textit{(ii)} employing recurrent or autoregressive architectures, when unrolled for a large enough number of steps;
\textit{(iii)} allowing explicit $n$-dependence in nonlinearities, though this departs from standard machine learning practice; and
\textit{(iv)} incorporating global non-scalar operations such as Slater determinants, which result in an effective $k\sim\mathrm{poly}(n)$.
As a concrete example, Ref.~\cite{PhysRevLett.131.036502} learned the volume-law ground state of the quantum
Sherrington-Kirkpatrick model using an MLP-NQS with $k\sim n$ neurons. This is consistent with our upper bound, and from our perspective, the choice $k\sim n = \Omega(n/\log n)$ was necessary.

\par 
Rigorous characterizations of NQS expressivity have appeared in restricted settings~\cite{deng2017quantum,sharir2022neural,gao2017efficient,sun2022entanglement,luo2023infinite,yang2024can}, alongside numerous empirical studies of expressivity~\cite{szabo2020neural,park2022expressive,PhysRevLett.134.079701,PhysRevLett.131.036502,westerhout2020generalization}. Our result complements these works by providing a largely architecture-agnostic upper bound on entanglement capacity, an aspect of expressivity, and capturing how the number of nonlinearities acts as a fundamental bottleneck.
\par 
Our work suggest several interesting directions for the future. For example, we expect our analyticity assumptions can be relaxed due to similar numerical results on ReLU networks~\cite{supp}. In the future, it will be interesting to find upper bounds on expressivity of NQS in terms of other quantities besides entanglement. Moreover, our work raises the possibility that NQS with only $O(1)$ nonlinearities and correspondingly small computational scaling could efficiently capture quantum critical states or quantum dynamics to times $t\sim\log n$~\cite{supp}.
\par 

\begin{acknowledgments}
\textit{Acknowledgments.} We thank Liang Fu, Filippo Gaggioli, Max Geier, and Di Luo for helpful discussions. We thank Di Luo for sharing details of a related work. We acknowledge support
from the Walter Burke Institute for Theoretical Physics at Caltech. 
\end{acknowledgments}

\bibliography{ref}

\end{document}


\title{Supplemental Material: Bound on entanglement in neural quantum states}

\author{Nisarga Paul}
\affiliation{Department of Physics, Massachusetts Institute of Technology, Cambridge, Massachusetts 02139, USA}
\affiliation{Department of Physics and Institute for Quantum Information and Matter,
California Institute of Technology, Pasadena, California 91125, USA}

\maketitle
\onecolumngrid
\tableofcontents

\section{Overview}
An overview of this Supplemental Material is as follows.

\par 
\vspace{0.5cm}

\noindent \textbf{Numerical experiments. } In \textsection \ref{sec:numerics}, we study the entanglement entropy of various neural quantum state (NQS) ans\"atze for subregions $A$ of $n$ spins. We divide our study into three categories: single-nonlinearity NQS (SN-NQS), multilayer perceptron NQS (MLP-NQS), and transformer NQS (T-NQS). The SN-NQS provides perhaps the simplest setting to illustrate our main claims: that entropy is bounded by $O(k\log n)$ and, in particular, volume law entanglement cannot be reached by a family of NQS given a bounded $k\sim O(1)$ number of elementary scalar nonlinearities. Numerical results show that logarithm law entanglement is generic for SN-NQS, consistent with this claim. Moreover, we provide an analytical proof that logarithm law entanglement is attained
for a simple example of SN-NQS, the Dicke state (Theorem \ref{thm:hh}). Numerical results also show that the MLP-NQS and T-NQS are consistent with logarithm law entanglement for a small number of nonlinearities. Finally, we provide details on the numerical figures presented in the main text.

\par 
\noindent \textbf{Proof of entanglement bound for NQS with a single nonlinearity. } In \textsection \ref{sec:proofsingle}, as a warmup to the general case, we prove that SN-NQS with certain analyticity assumptions satisfy an entanglement bound which prohibits volume law entanglement (Theorem \ref{thm:1nqs}). This introduces all the main elements of our general proof in a simpler setting.

\par 
\noindent 
\textbf{Feature reduction lemma. } In \textsection\ref{sec:frl}, we prove that the amplitudes of a feed-forward NQS acting on $n$ spins with $k$ nonlinearities is a function of at most $k+1$ affine features of the spins (Lemma \ref{lem:featurereduction}), after formalizing the computational graph of the NQS as a directed acyclic graph. This also provides a construction for the function $\mathcal{G}$ introduced in the main text. \par

\noindent 
\textbf{Proof of entanglement bound for NQS with $k$ nonlinearities. } In \textsection\ref{sec:proofO1}, we use the feature reduction lemma and provide the general proof, under certain analyticity assumptions, of our main result $S = O(\mu\log n)$ with $\mu\leq k+1$ (Theorem \ref{thm:O1nqs}). \par 

\noindent 
\textbf{Neural quantum state variational Monte Carlo. } In \textsection \ref{sec:nqsvmc}, we review how NQS is paired with variational Monte Carlo to solve ground state properties of quantum systems. In this context, we highlight how NQS ans\"atz with a bounded number $k\sim O(1)$ of nonlinearities can in principle be optimized more efficiently than general NQS by a factor poly($n$). \par 

\noindent 
\textbf{Other lemmas. } We leave several technical lemmas to \textsection\ref{sec:lems}, including the results from classical approximation theory on which our bound relies. \par 

As a technical remark, by an NQS we will mostly mean a \textit{family} of NQS labelled by $n$, where $n$ is the number of input spins. Asymptotic notation is then with respect to scaling in $n$. 

\section{Numerical experiments}\label{sec:numerics}

\subsection{Entanglement in single-nonlinearity neural quantum states}\label{sec:logsingle}

This section is devoted to a study of neural quantum states with a single, global nonlinearity, as a simple limit of the general case. In particular we consider states of $n$ qubits of the form 
\begin{subequations}\label{eq:PsiSingleNonlearity}
\begin{align}
    \ket{\Psi}&= \sum_{\vec s} \Psi(\vec s)\ket{\vec s}\quad \in \mathbb{C}^{2^n}\\
    \Psi(\vec s) &= \mathcal{N}^{-\frac12}\exp(\sigma(\vec w^T\vec s+b))\qquad \text{or} \qquad \Psi(\vec s) = \mathcal{N}^{-\frac12}\sigma(\vec w^T\vec s+b)
\end{align}
\end{subequations}
where $\vec s = (s_1,\ldots, s_n)$, $\ket{\vec s}$ is a $Z$-basis product state, $\vec w \in \mathbb{R}^n$, $b \in \mathbb{R}$, and $\sigma:\mathbb{R}\to \mathbb{C}$ is the nonlinearity. $\mathcal{N}^{-1/2}$ is the normalization. We illustrate this schematically in Fig.~\ref{fig:dicke}a. We refer to these as single-nonlinearity neural quantum states (SN-NQS). The two choices are technically equivalent under $\sigma \to \exp(\sigma(\cdot))$, but we keep these two parameterizations separate for convenience. Moreover, our results do not change appreciably if we had instead chosen $\Psi = \Psi_1+i\Psi_2$ as our parameterization. For the first case, when $\sigma$ is linear, $\Psi$ is a product state. To see this, let $\sigma(x) = a_0x+a_1$ and note that 
\begin{equation}
    \ket{\Psi} \propto \otimes_{i=1}^n \sum_{s_i=\pm 1} e^{a_0 w_is_i}\ket{s_i}.
\end{equation}
This shows that true nonlinearity of $\sigma$ is necessary for nontrivial entanglement to be present in this case. \par 

First, we will prove that there exists a SN-NQS such that $\ket\Psi$ has logarithmic entanglement. We refer to this state as the \textit{Dicke state}, in accordance with the literature~\cite{bartschi2019deterministic}. Next, we will explore the entanglement properties of SN-NQS for various choices of $\vec w,b,\sigma$ and provide evidence that logarithmic entanglement is generic. In Sec.~\ref{sec:proofsingle}, we will prove a tight log law entanglement scaling upper bound for this class of states. \par 

\begin{theorem}\label{thm:hh}
    (Dicke state) Take $\ket{\Psi}$ of the form Eq.~\eqref{eq:PsiSingleNonlearity} with $\Psi(\vec s) = \mathcal{N}^{-\frac12}\sigma(\vec w^T\vec s+b)$, $w_i =1, b=0,$ and $\sigma(x) = \delta_{x,0}$, and $n$ even. In practice, $\sigma(x)$ does not need to be discontinuous; $\sigma(x)$ only needs to agree with $\delta_{x,0}$ on the integers. Then $\ket{\Psi}$ has logarithmic entanglement scaling. In particular, its von Neumann entanglement entropy on any subset $A\subset \{1,\ldots, n\}$ satisfies $S_A = \frac12 \log |A| + O(1)$. 
\end{theorem}
\begin{proof}
Let $h(\vec s) =\sum_i \delta_{s_i,+1}$ be the Hamming weight of the string $\vec s$. Then the resulting normalized state is an equal-weight superposition over all states with half spins up, i.e. Hamming weight $n/2$: 

\begin{equation}
    \ket{\Psi} = \frac{1}{\sqrt{c_n}}\sum_{\substack{\vec s:\\h(\vec s) = n/2}} \ket{\vec s}
\end{equation}
where $c_n = \binom{n}{n/2}$. Let $A\subseteq \{1,\ldots, n\}$ be a subset of spins with $|A| = m$. Let $B$ be its complement with $|B| = n-m$. We decompose each bitstring as $\vec s = (\vec x, \vec y)$. The reduced density matrix on $A$ is 
\begin{equation}
    \rho_A(\vec x,\vec x') = \frac{\delta_{h(\vec x),h(\vec x')}}{c_n} \binom{n-m}{n/2 - h(\vec x)}.
\end{equation}
Hence $\rho_A$ is block-diagonal in $h(\vec x)$, and moreover each block is a rank$-1$ matrix of constants. The nonzero eigenvalues of $\rho_A$ can be read off as 
\begin{equation}
    \lambda_i = \frac{1}{c_n}\binom{n-m}{n/2 - i} \binom{m}{i} ,\qquad i = 1,\ldots, m.
\end{equation}
This is the hypergeometric distribution. For large $n$ and fixed $p \equiv m/n$, we obtain a binomial distribution, and the standard result for the entropy of the binomial distribution then gives $S_A =\frac12 \log(2\pi e \frac{n}{2}p(1-p)) + O(1/n) = \frac12 \log |A| + O(1)$. 
\end{proof}

We numerically calculate the von Neumann entanglement entropy for the Dicke state and plot the results in Fig.~\ref{fig:dicke}b. This affirms the logarithmic scaling and indicates the size of subleading corrections, which is useful for comparison with later finite-size numerics. We also note that for the Dicke state, the nonlinearity need only agree with $\delta_{x,0}$ on the integers. For example, it can be achieved by
\begin{equation}
    \sigma(x) = \relu(x-1) -2\relu(x) +\relu(x+1)
\end{equation}
in terms of $\relu(x) = \max(0,x)$, and therefore $\sigma$ need not be discontinuous. \par 

While the choices of $\vec w, b,$ and $\sigma$ for the Dicke state are convenient for proving Theorem \ref{thm:hh}, they are not necessary in general. In fact, SN-NQS \textit{generically} have at most logarithmic entanglement entropy scaling.

\begin{figure}[t]
    \centering
    \begin{minipage}[t]{0.42\linewidth}
        \centering
        \begin{tikzpicture}[x=1.0cm, y=0.6cm, >=stealth, node distance=1cm]
          \tikzstyle{neuron}=[circle,draw=black,minimum size=8pt,inner sep=0pt]
          \tikzstyle{input}=[neuron,fill=white!15]
          \tikzstyle{nonlin}=[neuron,fill=green!15]

          \node[input] (I1) at (0,-0.8) {};
          \node[input] (I2) at (0,-1.6) {};
          \node[input] (I3) at (0,-2.4) {};
          \node at (0,-2.9) {$\vdots$};
          \node[input] (I4) at (0,-3.8) {};
          \node[input] (I5) at (0,-4.6) {};
          \node[input] (I6) at (0,-5.4) {};
          \node at (0,-0.3) {$s_1,\ldots,s_n$};

          \node[nonlin] (SIG) at (1.8,-3.2) {$\sigma$};

          \node (PsiTxt) at (3.6,-3.2) {$\Psi(\vec s)$ or $\log \Psi(\vec s)$};

          \foreach \inp in {I1,I2,I3,I4,I5,I6}{
            \draw[-,thin] (\inp) -- (SIG);
          }
          \draw[-,thin] (SIG) -- (PsiTxt);
        \end{tikzpicture}
    \end{minipage}
    \begin{minipage}[t]{0.32\linewidth}
        \centering
        \includegraphics[width=\linewidth]{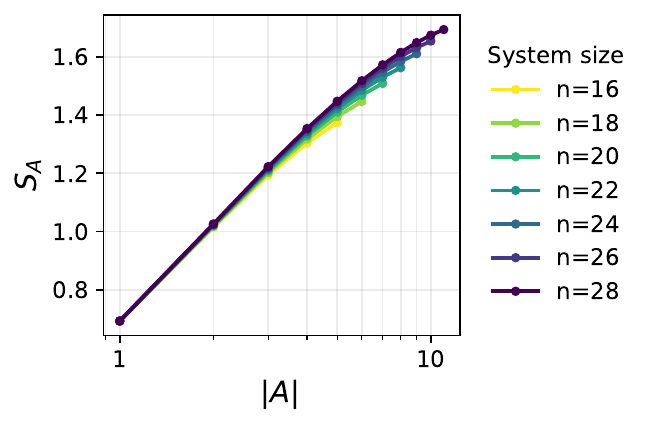}
    \end{minipage}
    \caption{(Left) Single-nonlinearity neural quantum state schematic. (Right) Von Neumann entanglement entropy $S_A$ for the Dicke state on a logarithmic $x$-axis. $n$ is the total number of qubits.}
    \label{fig:dicke}
\end{figure}

For now, to support this claim, we numerically explore the entanglement entropy for several random ensembles of states of the form Eq.~\eqref{eq:PsiSingleNonlearity} in Fig.~\ref{fig:nqs_single}. In all cases, we considered $n=22$, considered $20$ trial wavefunctions, and plotted the resulting mean and standard deviation. In the following, we write $\sigma(x) = \sigma_1(x)+i\sigma_2(x)$ where $\sigma_{1,2}$ are the real and imaginary parts, respectively. 
\par 
The first ensemble is the ``real" ensemble, where $\sigma_2 = 0$ and we choose among $\sigma_1\in \{\tanh, \sin, \relu, \mathrm{GELU}\}$. The GeLU (gaussian error linear unit) ~\cite{hendrycks2016gaussian} is a smooth approximation to $\relu$. These states have positive, real amplitudes. We choose the weights and bias from a normal distribution $\mathcal{N}(0,1)$. We plot the entanglement scaling in Fig.~\ref{fig:nqs_single}a. The bounded nonlinearities (tanh, sin) show nonzero entanglement consistent with a log law, while the unbounded nonlinearities (ReLU, GELU) show near-vanishing entanglement. This can be attributed to the following. The amplitudes for the former case are bounded between $[1/e,e]$ while the amplitudes for the latter case in principal have no upper bound, leading to typical configurations in which the amplitudes localize to one or a few basis states, with little or no entanglement. \par 

The second ensemble is the ``pure phase" ensemble, where $\sigma_1 = 0$ and amplitudes are pure phases $\exp(i\sigma_2(\cdots))$. Once again, we choose among $\sigma_2 \in \{\tanh,\sin,\relu, \mathrm{GELU}\}$ and weights and bias from $\mathcal{N}(0,1)$. We plot the entanglement scaling in Fig.~\ref{fig:nqs_single}b. In this case, all ensembles of states show nonzero entanglement consistent with a log law. This is consistent with the previous observation that states with amplitudes bounded from below and above have more entanglement. \par 

The third ensemble is the ``general" ensemble, for which nonlinearities were chosen to be $(1+i)\sigma(t)$ with $\sigma \in \{\tanh,\sin,\relu, \mathrm{GELU}\}$ and the weights and bias chosen from $\mathcal{N}(0,1)$ as before. We plot the entanglement scaling in Fig.~\ref{fig:nqs_single}c. This case bears the closest resemblance to the real case, especially with regards to the near-vanishing entanglement for the unbounded ReLU and GELU nonlinearities, which give rise to concentration of amplitude magnitudes on a few basis states. Due to the presence of complex phases, the entanglement achieved by the other nonlinearities is slightly boosted; for instance, $S_A$ reaches $\sim 1.25$ instead of $\sim 0.8$ for the $\sin$ nonlinearity. 
\begin{figure}
    \centering
    \includegraphics[width=0.32\linewidth]{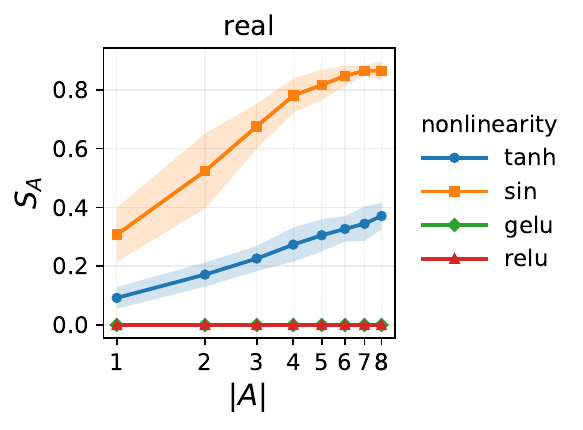}   \includegraphics[width=0.32\linewidth]{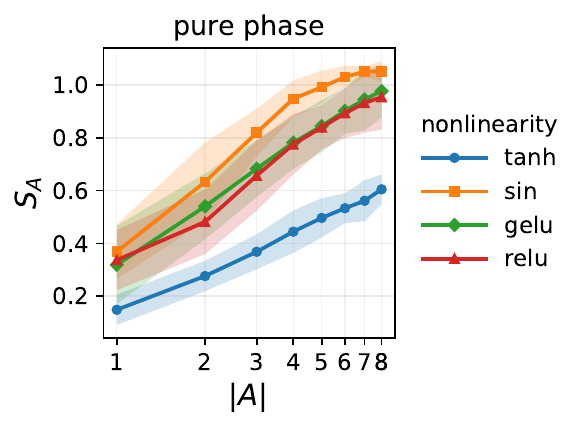}        \includegraphics[width=0.32\linewidth]{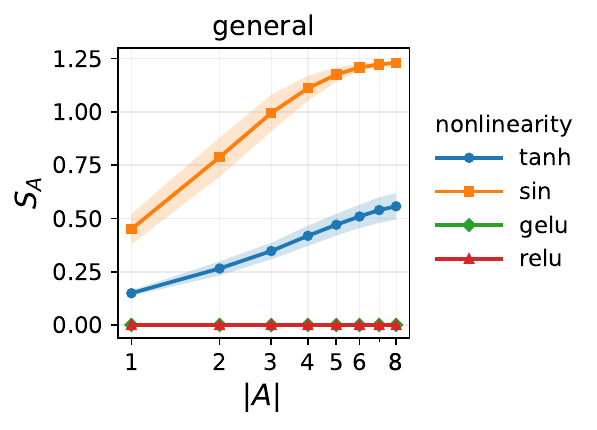}  
        \caption{Von Neumann entanglement entropy $S_A$ for three ensembles of single nonlinearity neural quantum states (SN-NQS) acting on $n=22$ qubits, plotted versus subsystem size $|A|$ (logarithmic $x$-axis). (Left) SN-NQS of the form Eq.~\eqref{eq:PsiSingleNonlearity} with $\sigma$ given by the indicated nonlinearity and $\vec w,b$ drawn from $\mathcal{N}(0,1)$, so the amplitudes are all real. (Center) The same, but with $\sigma$ given by $i$ times the indicated nonlinearity so the amplitudes are pure phase. (Right) The same, except $\sigma$ given by $1+i$ times the indicated nonlinearity so the amplitudes are general complex numbers. Shaded color indicates one standard deviation.}
    \label{fig:nqs_single}
\end{figure}

\subsection{Entanglement in multilayer perceptron neural quantum states}\label{sec:MLPNQS}

In this section we explore entanglement in multilayer perceptron (MLP) neural quantum states. Our MLP-NQS is built by feeding spin
configurations $\vec s \in \{\pm 1\}^n$ through a vanilla feed-forward
MLP of width $w$ and depth $d$. Formally, this is constructed as follows. For each layer 
$\ell=1,\ldots,d$, define
\begin{equation}
    \vec z^{(\ell)} = W^{(\ell)} \vec h^{(\ell-1)} + \vec b^{(\ell)}, 
    \qquad \vec h^{(\ell)} = \sigma\!\left(\vec z^{(\ell)}\right),
\end{equation}
where $W^{(\ell)} \in \mathbb{R}^{w \times d_{\ell-1}}$ is the weight matrix, 
$\vec b^{(\ell)} \in \mathbb{R}^w$ is the bias vector, and $\sigma$ is either 
$\tanh$ or $\sin$. We take $\vec h^{(0)} \equiv \vec s$, so that the input 
dimension is $d_0 = n$. The final feature vector is 
$\vec h(\vec s) \equiv \vec h^{(d)} \in \mathbb{R}^w$. The wavefunction is then 
obtained by a complex affine projection of this feature vector:
\begin{equation}
    \Psi(\vec s) = \left(\vec w_\mathrm{R}^T \vec h(\vec s) + b_\mathrm{R}\right) 
    + i\left(\vec w_\mathrm{I}^T \vec h(\vec s) + b_\mathrm{I}\right),
\end{equation}
with trainable parameters $(\vec w_\mathrm{R}, b_\mathrm{R})$ and 
$(\vec w_\mathrm{I}, b_\mathrm{I})$. An example is illustrated schematically in Fig.~\ref{fig:mlp-nqs}. Weights are initialized with independent Gaussian entries
\begin{equation}
    W^{(\ell)}_{ij} \sim \mathcal{N}\!\left(0,\sigma_w^2\right), 
    \qquad b^{(\ell)}_i \sim \mathcal{N}(0,\sigma_b^2).
\end{equation}
Unless otherwise specified, we take $\sigma_w\to \sigma_w/\sqrt{\mathrm{fan\mbox{-}in}}$, where fan-in denotes the number of input connections to the 
layer. This variance scaling preserves stable propagation of activations 
through the depth of the network~\cite{lecun2002efficient}. The same initialization scheme is applied to 
the final affine heads $(\vec w_\mathrm{R},b_\mathrm{R})$ and 
$(\vec w_\mathrm{I},b_\mathrm{I})$. A schematic is shown in Fig.~\ref{fig:mlp-nqs}. \par 

To assess entanglement properties, we construct the full normalized statevector 
$\Psi(\vec s)$ over all $2^n$ configurations and evaluate the von Neumann 
entropy $S_A$ of subsystems $A$ with $|A|=m$. For each subsystem size $m$ we 
consider $n_\mathrm{trans}=5$ randomly chosen contiguous windows 
$\{j,j+1,\ldots,j+m-1\}$ (with periodic boundary conditions) and average the 
resulting entropies. This limited translation averaging reduces computational 
cost while still probing typical entanglement behavior. We chose $\sigma_w =1.0, \sigma_b = 0.2$. For varying choices of $w,d$ we generate $20$ independent random trial 
wavefunctions and plot mean and standard deviation across the ensemble in Figs.~\ref{fig:MLP1}. In Fig.~\ref{fig:mlp-nqs2} we vary the width, holding other hyperparameters constant, and also vary $n$ for fixed $w=5,d=2$ and $|A|=n/2$. We chose to omit ReLU and other unbounded nonlinearities, which typically exhibited very low entanglement (near zero) in all cases we studied for the MLP-NQS; we ascribe this to amplitude blowups which localized the state around one or a few product basis states. 

\begin{figure}
    \centering
\begin{tikzpicture}[x=1.3cm, y=0.5cm, >=stealth, node distance=0.5cm]
  \tikzstyle{neuron}=[circle,draw=black,minimum size=8pt,inner sep=0pt]
  \tikzstyle{input}=[neuron,fill=white!15]
  \tikzstyle{hidden}=[neuron,fill=green!15]
  \tikzstyle{output}=[neuron,fill=red!15]

          \node[input] (I1) at (0,-1.0) {};
          \node[input] (I2) at (0,-2) {};
          \node[input] (I3) at (0,-3) {};
          \node at (0,-3.9) {$\vdots$};
          \node[input] (I4) at (0,-5) {};
          \node[input] (I5) at (0,-6) {};
          \node[input] (I6) at (0,-7) {};
          \node at (0,-0.3) {$s_1,\ldots,s_n$};

  \foreach \i in {1,...,4}{
    \node[hidden] (H1\i) at (1.5,-\i-1.4) {};
  }

  \foreach \i in {1,...,4}{
    \node[hidden] (H2\i) at (3.0,-\i-1.4) {};
  }

  \node[output] (R) at (4.7,-3.4) {$\Re$};
  \node[output] (I) at (4.7,-4.4) {$\Im$};
  \node (Psi) at (6.2,-3.9) {$\Psi(\vec s)$ or $\log \Psi(\vec s)$};

  \foreach \i in {1,...,6}{
    \foreach \j in {1,...,4}{
      \draw[-,thin] (I\i) -- (H1\j);
    }
  }
  \foreach \i in {1,...,4}{
    \foreach \j in {1,...,4}{
      \draw[-,thin] (H1\i) -- (H2\j);
    }
  }
  \foreach \i in {1,...,4}{
    \draw[-,thin] (H2\i) -- (R);
    \draw[-,thin] (H2\i) -- (I);
  }
  \draw[-,thin] (R) -- (Psi);
  \draw[-,thin] (I) -- (Psi);

\end{tikzpicture}
    \caption{Schematic of an multilayer perceptron neural quantum state (MLP-NQS) on $n$ spins. Here, the input spin configuration
$\vec s$ is processed by $d=2$ hidden layers of width $w=4$ with 
nonlinear activations. The final feature vector is mapped 
affinely to two outputs, giving the real and imaginary parts of the amplitude 
$\Psi(\vec s)$.}
    \label{fig:mlp-nqs}
\end{figure}
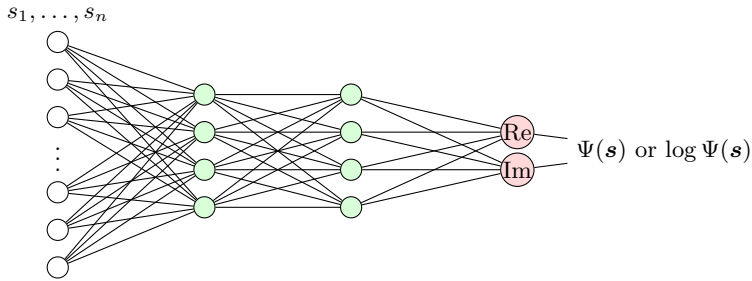

\begin{figure}
    \centering
    \includegraphics[width=0.32\linewidth]{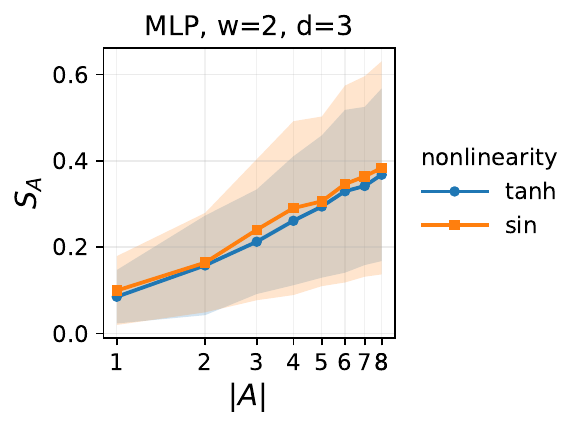}
    \includegraphics[width=0.32\linewidth]{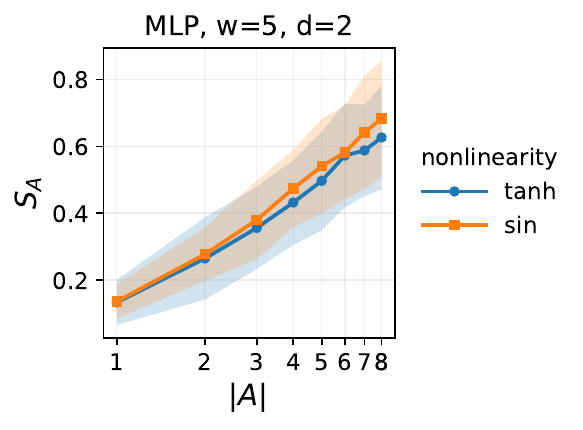}
    \includegraphics[width=0.32\linewidth]{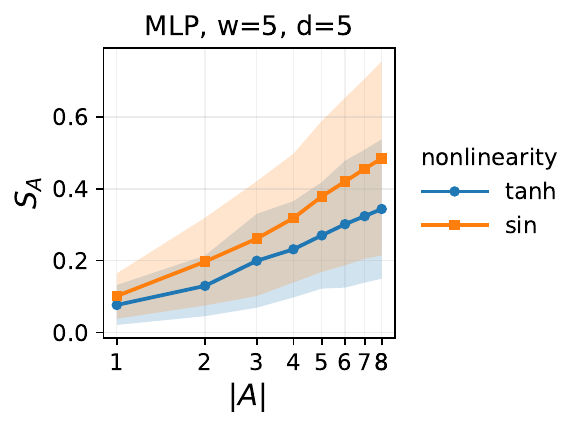}
    \caption{Von Neumann entanglement entropy for multilayer perceptron neural quantum states (MLP-NQS) with width $w$, depth $d$, and total number of nonlinearities $k=wd$, as described in Sec.~\ref{sec:MLPNQS}, acting on $n=22$ qubits, plotted versus subsystem size $|A|$ (logarithmic $x$-axis). (Left) $w=2,d=3$, (Center) $w=5,d=2$, (Right) $w=5,d=5$. In all cases $k \lesssim n$ and entanglement scaling is logarithmic, consistent with our bound. Shaded color indicates one standard deviation.}
    \label{fig:MLP1}
\end{figure}

\begin{figure}
\centering
\includegraphics[width=0.32\linewidth]{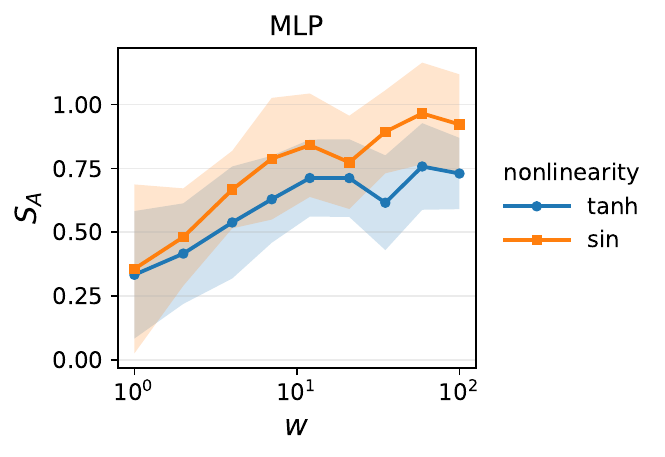}
\includegraphics[width=0.31\linewidth]{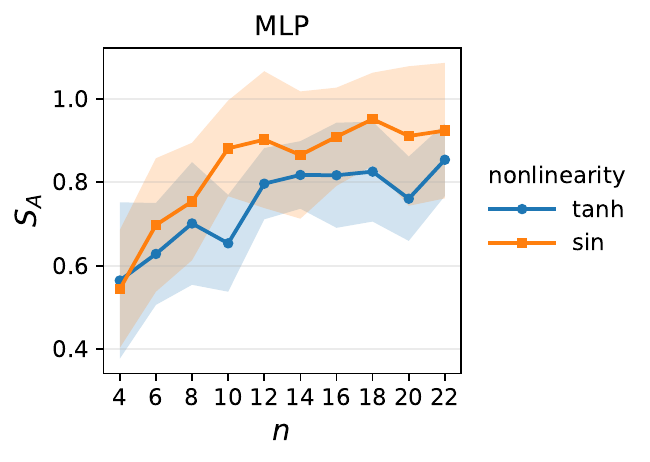}
\caption{(Left) Von Neumann entanglement entropy for MLP-NQS on $n=22$ spins for $d=2$ and varying width $w$ (logarithmic $x$-axis). This indicates entanglement increases sublinearly in the number of nonlinearities $k=wd$ for this ans\"atz, consistent with our bound. Shaded color indicates one standard deviation. (Right) Varying $n$ for fixed $|A|=n/2$, $w=5,d=2$. }
\label{fig:mlp-nqs2}
\end{figure}

\subsection{Entanglement in transformer neural quantum states}

\begin{figure}[t]
\centering
\begin{tikzpicture}[x=1.6cm, y=0.5cm, >=stealth, node distance=1cm]
  \tikzstyle{neuron}=[circle,draw=black,minimum size=8pt,inner sep=0pt]
  \tikzstyle{input}=[neuron,fill=white!15]
  \tikzstyle{block}=[draw=black,rounded corners,minimum width=22mm,minimum height=9mm,align=center,fill=green!12]
  \tikzstyle{blockghost}=[draw=black,rounded corners,minimum width=22mm,minimum height=9mm,align=center,fill=green!8]
  \tikzstyle{connector}=[draw=black,rounded corners,fill=green!8,minimum height=8mm,minimum width=8mm,align=center]
  \tikzstyle{output}=[neuron,fill=red!15]
  \node at (0,0.1) {$s_1,\ldots,s_n$};
  \node[input] (S1) at (0,-0.8) {};
  \node[input] (S2) at (0,-1.6) {};
  \node[input] (S3) at (0,-2.4) {};
  \node[input] (S4) at (0,-3.2) {};
  \node[input] (S5) at (0,-4.0) {};
  \node[input] (S6) at (0,-4.8) {};
  \node[input] (S7) at (0,-5.6) {};
  \node at (0,-6.4) {$\vdots$};
  \draw[rounded corners=3pt,thin] (-0.15,-0.4) rectangle (0.15,-2.8);  
  \draw[rounded corners=3pt,thin] (-0.15,-2.0) rectangle (0.15,-4.4);  
  \draw[rounded corners=3pt,thin] (-0.15,-3.6) rectangle (0.15,-6.0);  
  \node (P1) at (0.15,-1.6) {};
  \node (P2) at (0.15,-3.2) {};
  \node (P3) at (0.15,-4.8) {};
  \node[anchor=west] at (0.4,-1.6) {\small tokens};
  
  \node[blockghost, minimum width=20mm] (ATTN1) at (1.65,-3.40) {$\mathcal{A}$};
  \node[blockghost, minimum width=20mm] (ATTN2) at (1.75,-3.30) {$\mathcal{A}$};
  
  \node[block, minimum width=12mm] (FFN) at (3.00,-3.20) {FFN};
  
  \draw[-,thin] (P1) -- (ATTN1.west);
  \draw[-,thin] (P2) -- (ATTN1.west);
  \draw[-,thin] (P3) -- (ATTN1.west);
  \draw[-,thin] (P1) -- (ATTN2.west);
  \draw[-,thin] (P2) -- (ATTN2.west);
  \draw[-,thin] (P3) -- (ATTN2.west);
  
  \draw[-,thin] (ATTN1.east) -- (FFN.west);
  \draw[-,thin] (ATTN2.east) -- (FFN.west);
  
  \node[block, minimum width=20mm] (ATTN3) at (1.85,-3.20) {$\mathcal{A}$\\[-2pt]\footnotesize self-attention};
  
  \draw[-,thin] (P1) -- (ATTN3.west);
  \draw[-,thin] (P2) -- (ATTN3.west);
  \draw[-,thin] (P3) -- (ATTN3.west);
  
  \draw[-,thin] (ATTN3.east) -- (FFN.west);
  
  \node[connector] (CAT) at (3.90,-3.20) {$\mathrm{concat}$};
  \node[output] (R) at (4.60,-2.60) {$\Re$};
  \node[output] (I) at (4.60,-3.80) {$\Im$};
  \node (PsiTxt) at (6.2,-3.20) {$\Psi(\vec s)$ or $\log \Psi(\vec s)$};
  \draw[-,thin] (FFN) -- (CAT);
  \draw[-,thin] (CAT) -- (R);
  \draw[-,thin] (CAT) -- (I);
  \draw[-,thin] (R) -- (PsiTxt);
  \draw[-,thin] (I) -- (PsiTxt);
\end{tikzpicture}
\caption{Transformer neural quantum state. Eight input spins are grouped into overlapping 3-spin patches with stride 2 (overlap 1). Patches are treated as tokens ($x_j=p_j$) and passed to a multi-head self-attention block $\mathcal{A}$ (shown as overlapping rounded rectangles), followed by a feed-forward nonlinearity. Contextualized tokens are concatenated and mapped to real/imaginary heads, yielding $\Psi(\vec s)$ in phase--magnitude form.}
\label{fig:transformer-nqs}
\end{figure}
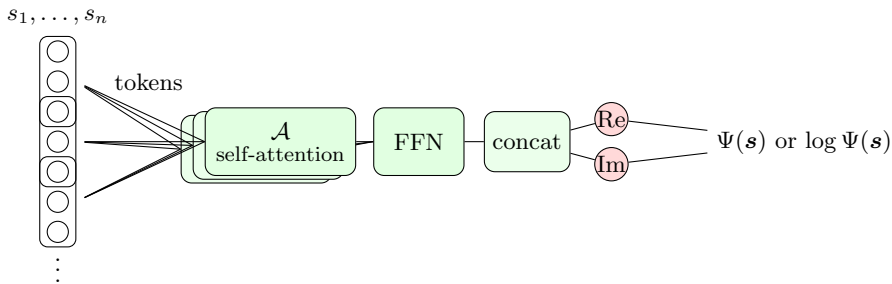

We now consider a transformer neural quantum state (T-NQS). These are related to the NQS ans\"atz introduced in Ref.~\cite{viteritti2023transformer}, which proved quite effective in capturing ground states of frustrated spin models.
Given a spin configuration $\vec s$ of length $n$, we divide the chain into overlapping local patches of fixed size $P$ with stride $\Delta$, producing 
\[
p_j = (s_{j},s_{j+1},\ldots,s_{j+P-1}), \qquad j=1,\ldots,M,
\]
where $M = \lfloor (n-P)/\Delta \rfloor + 1$ is the number of patches. We choose $\Delta = O(n)$ so that there are $O(1)$ patches in total. In general one could embed each patch $p_j$ into a higher-dimensional vector $x_j \in \mathbb{R}^d$, 
but for simplicity we take the embedding to be the identity, $x_j = p_j$, in what follows, and describe the input to the model as the sequence of tokens $X = (x_1,\ldots,x_M)$. This sequence of tokens is processed by self-attention heads. For each head and 
for each token $x_j$, queries, keys, and values are formed as
\begin{equation}
Q = X W^Q, \qquad K = X W^K, \qquad V = X W^V,
\end{equation}
with matrices $W^Q,W^K,W^V \in \mathbb{R}^{P \times d}$ and attention weights are computed as
\begin{equation}
A = \mathrm{softmax}\!\left(\frac{QK^T}{\sqrt{d}}\right).
\end{equation}
Finally, the outputted tokens are obtained as $Z = A V$. After $L$ layers of this procedure, each token $z_j$ is then passed through a feed-forward network of the form
\begin{equation}
z_j \mapsto W_2 \, \sigma(W_1 z_j + b_1) + b_2,
\end{equation}
where $\sigma$ is a scalar nonlinearity. 
The outputs are concatenated to a single feature vector $\mathrm{vec}(Z) \in \mathbb{R}^{Md}$, 
which is linearly projected to two real scalars, $f_R(\vec s)$ and $f_I(\vec s)$. 
The wavefunction amplitude is then as
\begin{equation}
\Psi(\vec s) \quad \text{or} \quad \log \Psi(\vec s) =  f_R(\vec s) + i f_I(\vec s).
\end{equation}
This is illustrated schematically in Fig.~\ref{fig:transformer-nqs}. Weights are chosen as $W_{ij} \sim \mathcal{N}\!\left(0,\sigma_w^2\right),
$ and biases are drawn as $b_i \sim \mathcal{N}(0,\sigma_b^2)$, as before. For plots of entanglement in T-NQS we refer the reader to the main text. 

\subsection{Numerical details} 

Here we report numerical details of the generation of Figure 1 in the main text. For the SN-NQS, we used an ans\"atz of the form $\Psi(\vec s) \propto \exp(i\tanh(\vec w^T\vec s + b))$, \textit{i.e.} a pure phase model with one scalar tanh nonlinearity acting on one affine feature. We restricted to the pure phase model because of the larger entropy relative to the real and general models (c.f. \textsection \ref{sec:logsingle}) though we have shown that other choices behave similarly. Parameters were chosen as $w_i \sim \mathcal{N}(0,1.0), b\sim \mathcal{N}(0.0.5^2)$. For each point we generated 20 independent trial wavefunctions and averaged over 10 random bipartitions or subregions. For the MLP-NQS, our ans\"atz had width $w=3$, depth $d=2$, tanh nonlinearities, and fixed deterministic output projections (all weights set to $1$). We chose to keep the latter fixed to reduce variance. The weights of the MLP were initialized to $w_i \sim \mathcal{N}(0,1.0), b\sim \mathcal{N}(0,0.2^2)$. A LayerNorm was used before each nonlinearity. For each point, we generated 20 trial wavefunctions and averaged over 10 random biparitions or subregions. For the T-NQS, we used a patch size $P=6$, stride 5, embedding dimension 32, 4 attention heads, and 4 layers. The following feed-forward blocks used tanh nonlinearities with width 64 and embedding and attention projections were frozen to reduce variance. That is, only the feed-forward blocks varied across trials with weights initialized to $w_i \sim \mathcal{N}(0,1.0), b\sim \mathcal{N}(0,0.2^2)$. Finally, the output projection heads were fixed, deterministic projections and amplitudes were taken in phase-magnitude form. For each point, we took 20 trial wavefunctions and 10 random bipartitions or subregions.

\section{Proof of entanglement bound for NQS with a single nonlinearity} \label{sec:proofsingle}

For completeness, we repeat the setup. Consider a state on $n$ spins, defined by
\begin{equation}
    \ket{\Psi} = \sum_{\vec s} \Psi(\vec s) \ket{\vec s}
\end{equation}
where $\vec s = (s_1,s_1,\ldots, s_{n})$ labels the spins in e.g. the $Z$ eigenbasis, and $\Psi(\vec s) \in \mathbb{C}$. We assume $\langle \Psi |\Psi\rangle = 1$. We consider the case where $\ket{\Psi}$ is a neural quantum state (NQS) defined by a single nonlinearity for the sake of simplicity, and we generalize to more nonlinearities in Section \ref{sec:proofO1}. In particular, we choose a state of the form
\begin{equation}\label{eq:Psiofs}
    \Psi(\vec s) = \exp(\sigma(\vec w^T\vec s+b)).
\end{equation}
This is close to the most general form for a feed-forward neural network with a single nonlinearity. The most general form is $e^{\sigma(\vec w^T\vec s+b) + \vec u^T\vec s + c}$, but the extra factor will only rescale the rows / columns of the matrix $\widetilde M$ (to be defined below) by separate factors and will not change our conclusions. Hence the form we adopt is without loss of generality. Here $\sigma:\mathbb{R}\to \mathbb{C}$ is an \textit{arbitrary} nonlinearity subject to a few mild constraints due to our technical assumptions below, $w\in\mathbb{R}^n$ is a weight vector, and $b\in \mathbb{R}$ is a bias. We denote $\{1,\ldots, n\}$ by $[n]$. We refer to $t = \vec w^T\vec s+b$ as the \textit{pre-activation}, and when we wish to emphasize the function of $t$ alone we write 
\begin{equation}
    F(t) = \exp(\sigma(t)). 
\end{equation}
We consider entanglement entropy $S_A$ of subregions $A\subset [n]$. Finally, we note that a Bernstein ellipse is the complex domain 
\begin{equation}
    B(a) = \left\{z=x+iy \in \mathbb{C} : \frac{x^2}{\cosh^2(a)} + \frac{y^2}{\sinh^2(a)} < 1\right\}
\end{equation}
where $a>0$, which contains $[-1,1]$. 

\begin{theorem}\label{thm:1nqs}
Suppose $\ket{\Psi}$ is the NQS on $n$ qubits with $\Psi({\vec s})$ defined in Eq.~\eqref{eq:Psiofs}. Suppose $\sup_{\vec s}|\vec w^T\vec s+b|<\bar t$ for an $O(1)$ constant $\bar t>0$. Suppose the function $f(x) = \exp(\sigma(\bar t x))$ admits an analytic extension $f(z)$ to a Bernstein ellipse $B(a)$ enclosing $[-1,1]$ and satisfies $|f(z)| \leq C$ for some $O(1)$ constant $C>0$ on $B(a)$. Denote by $S_A$ the von Neumann entropy of the reduced density matrix on a subset $A\subset [n]$. Then 
    \begin{equation}
        S_A = O(\log n).
    \end{equation}
In particular, no single-nonlinearity NQS of this form can achieve volume law entanglement. 
\end{theorem}
\begin{proof}
The proof proceeds with the following steps. 
\begin{enumerate}
    \item We bound $S_A$ by $\log \rank M$, where $M$ is an auxiliary matrix whose entries are defined by $F(t)$.
    \item We approximate $F(t)$ on a bounded domain by a polynomial of degree $d = \Theta(n)$ with exponential accuracy. 
    \item We convert this into a bound on the rank of $M$ by $O(d^2)$.
\end{enumerate}

\textbf{Step 1. }
Let $A \subset [n]$ and let $\bar A$ be its complement. Let us refer to 
\begin{equation}
    z_s = t = \vec w^T\vec s+b
\end{equation}
interchangeably. Let $\mathcal{H}_A, \mathcal{H}_{\bar A}$ denote the Hilbert spaces on $A$ and $\bar A$ and let $\ket{\vec u}, \ket{\vec v}$ be $Z$-basis states on $A$ and $\bar A$, respectively. Let us define \textit{split affine features} $x_{u},y_{v} \in\mathbb{R}$ by $x_{u}= \sum_{i\in A} w_i u_i +b/2$, $y_v = \sum_{i\in \bar A} w_i v_i + b/2$, which satisfy $x_{ u}+y_{ v} = z_s$. Write
\begin{equation}
    \ket{\Psi} = \sum_{\vec u,\vec v} M_{\vec u,\vec v}\ket{\vec u}_A\ket{\vec v}_{\bar A},\qquad M_{\vec u,\vec v} =  F(x_u + y_v)
\end{equation}
where $F(t) = \exp(\sigma(t))$. Then
\begin{equation}
\begin{aligned}
    \rho_A &= \Tr_{\bar A}\ket{\Psi}\bra{\Psi} = \sum_{\vec u,\vec u',\vec v} M_{\vec u,\vec v}M^*_{\vec u',\vec v} \ket{\vec u}_A\bra{\vec u'}_A\\
    &= \sum_{\vec u,\vec u'} (MM^\dagger)_{\vec u,\vec u'}\ket{\vec u}_A\bra{\vec u'}_A. 
\end{aligned}
\end{equation}
The von Neumann entropy is $S_A = -\Tr(\rho_A\log \rho_A)$. We note that $S_A \leq \log \rank MM^\dagger$. To see this, denote the eigenvalues of $MM^\dagger$ by $\lambda_i\geq 0$ for $i=1,\ldots, r$ where $ r= \rank MM^\dagger$. Because the uniform distribution maximizes von Neumann entropy, $-\sum_{i=1}^r \lambda_i \log \lambda_i \leq \log r$. Furthermore, we note that $\rank MM^\dagger = \rank M$ and use the fact 
\begin{equation}
    S_A \leq \log \rank M
\end{equation}
for what follows. \par 

\textbf{Step 2. }  
We have assumed the pre-activation is bounded, $|t|<\bar t$ for some constant $\bar t>0$. Write $F(t) = \exp(\sigma(t))$ and define $f(x) = F(\bar tx)$. Our assumption is that $f(x)$ can be analytically extended to a function $f(z)$ on a Bernstein ellipse $B(a)$ neighborhood of $[-1,1]$ where $a>0$ and where it satisfies $|f(z)|<C$ for some $O(1)$ constant $C$. Then for $x \in [-1,1]$ we have
\begin{equation}
    \left|f(x) - \sum_{n=0}^d a_{n}T_n(x)\right| \leq \frac{2Ce^{-ad}}{e^a-1}
\end{equation}
for suitable coefficients $a_{n}$ satisfying $|a_{n}|<2Ce^{-an}$ for $n>0$, where $T_n(x)$ is the $n$th Chebyshev polynomial (Lemma \ref{lemma:bernstein}). It follows that 
\begin{subequations}
    \begin{align}
        \left|F(t) -P_d(t)\right| \leq  \frac{2Ce^{-ad}}{e^a-1} \equiv \varepsilon_{\mathrm{poly}}
    \end{align}
\end{subequations}
for $t \in [-\bar t,\bar t]$ with a degree $d$ polynomial $P_d(t) = \sum_{n=0}^d a_n T_n(t/\bar t)$. At a later stage, we will need to bound $2\sqrt{\varepsilon_{\mathrm{poly}}} 2^{n/4} n$ by an $O(1)$ constant. With this in mind, it suffices to take $d$ to satisfy 
\begin{equation}\label{eq:degree}
    d \geq \frac{n}{2a} \ln 2 + \frac1a \ln n + \frac1a \ln\left(\frac{8C}{e^a-1}\right) = \Theta(n)
\end{equation}
which implies $2\sqrt{\varepsilon_{\mathrm{poly}}} 2^{n/4} n \leq 1$. \par 

\textbf{Step 3. } From Step 2, there exists a degree $d$ polynomial $P_d:\mathbb{R}\to \mathbb{C}$, defined as $P_d(t) = \sum_{k=0}^d \alpha_k t^k$, such that 
\begin{equation}
    \max_{t \in [-\bar t,\bar t]} |F(t) -P_d(t)| \leq \varepsilon_{\text{poly}}
\end{equation}
with $d = \Theta(n)$. In what follows, we write
\begin{equation}
    F(z_s) = P_d(z_s) + \Delta_s
\end{equation}
and define $Q  = \|P_d\|_2^2 \equiv \sum_{\vec s}|P_d(z_s)|^2$. We similarly define $\|\Delta\|_2^2 = \sum_{\vec s} |\Delta_s|^2 \leq 2^{n}\varepsilon_{\mathrm{poly}}^2 $ and $\|F\|_2^2 = \sum_{\vec s} |F(z_s)|^2 = 1$. We define the normalized auxiliary state 
\begin{equation}
    \ket{\tilde \Psi} = Q^{-1/2} \sum_{\vec s} P_d(z_s) \ket{\vec s} = Q^{-1/2}\sum_{\vec u,\vec v} \widetilde M_{\vec u,\vec v} \ket{\vec u}_A\ket{\vec v}_{\bar A},\qquad \widetilde M_{\vec u,\vec v} =  P_d(x_u+y_v).
\end{equation}
Let $\tilde S_A$ be the von Neumann entropy of $\tilde \rho_A = \Tr_{\bar A}\ket{\tilde \Psi}\bra{\tilde \Psi}$, which again satisfies $\tilde S_A \leq \log\rank \widetilde M$, following identically the previous arguments for $S_A$. We can bound the rank of $\widetilde M$ as follows. Note that
\begin{subequations}
\begin{align}
    \widetilde M &= \sum_{\vec u,\vec v} \widetilde M_{\vec u,\vec v} \ket{\vec u}\bra{\vec v} \\
    &= \sum_{\vec u,\vec v}\sum_{k=0}^d\alpha_k(x_u+y_v)^k \ket{\vec u}\bra{\vec v}\\
    &= \sum_{\vec u,\vec v} \sum_{k=0}^d \sum_{m=0}^k \alpha_{k}\binom{k}{m} x_u^my_v^{k-m} \ket{\vec u}\bra{\vec v}\\
    &= \sum_{k=0}^d\sum_{m=0}^k \alpha_k\binom{k}{m}\left(\sum_{\vec u} x_u^m \ket{\vec u}\right)\left(\sum_{\vec v} y_v^{k-m} \bra{\vec v}  \right)\\
    &\equiv \sum_{k=0}^d\sum_{m=0}^k \tilde m_{k,m}^{\text{rank-1}}
\end{align}
\end{subequations}
where we've expanded using $P_d(t) = \sum_{k=0}^d \alpha_{k} t^k$ and the binomial theorem. This expresses $\widetilde M$ as the sum of rank-1 matrices $\tilde m_{k,m}^{\text{rank-1}}$. By the subadditivity of rank, 
\begin{equation}
  \rank  \widetilde M \leq \sum_{k=0}^d\sum_{m=0}^k1=\frac12(d+1)(d+2) = \Theta(n^2)
\end{equation}
which implies $\tilde S_A \leq \log \frac12(d+1)(d+2) = O(\log n)$ by Step 1. To transform this into a bound for $S_A$, we need a few more steps. \par 

First, we define $\ket{\delta} = \ket{\Psi} - \ket{\tilde \Psi}$ and wish to bound $\|\ket{\delta}\|_2^2$. We have
\begin{subequations}
    \begin{align}        \|\ket\delta\|_2^2 &= 2-2\Re \langle \Psi|\tilde \Psi\rangle \\
    &= 2(1-Q^{-1/2}\Re  \sum_{\vec s} F(z_s)^* P_d(z_s))\\
    &= 2(1-Q^{-1/2} \Re  \sum_{\vec s} |P_d(z_s)|^2  - Q^{-1/2} \Re\sum_{\vec s}\Delta_s^*P_d(z_s))\\
    &= 2(1-Q^{1/2}- Q^{-1/2} \Re\sum_{\vec s}\Delta_s^*P_d(z_s))\\
    &\leq 2 |1-Q^{1/2}| + 2\left| Q^{-1/2} \sum_{\vec s} \Delta_s^* P_d(z_s)\right|.
    \end{align}
\end{subequations}
We bound the two terms individually. First, using $\|F\|_2 \leq \|P_d\|_2 + \|\Delta\|_2$ we have $|1-Q^{1/2}| = | \|F\|_2-\|P_d\|_2| \leq \|\Delta\|_2$. Second, using the Cauchy-Schwarz inequality we have 
\begin{equation}
    \left| Q^{-1/2} \sum_s \Delta_s^* P_d(z_s)\right| \leq Q^{-1/2} \|\Delta\|_2 \|P_d\|_2 = \|\Delta \|_2. 
\end{equation}
We conclude 
\begin{subequations}
    \begin{align}        \|\ket\delta\|_2^2 &\leq 2\cdot 2 \cdot \|\Delta\|_2 \\
    &\leq 4\varepsilon_{\mathrm{poly}}2^{n/2}\\
    \to\|\ket\delta\|_2 &\leq 2\sqrt{\varepsilon_{\mathrm{poly}}}2^{n/4}.
    \end{align}
\end{subequations}
Next, we wish to bound the trace distance $\frac12 \|\rho_A-\tilde \rho_A\|_1$. From Lemma \ref{lemma:fvdg} we have 
\begin{equation}
    \frac12 \|\rho-\tilde \rho\|_1 \leq \|\ket\delta\|_2 \leq 2\sqrt{\varepsilon_{\text{poly}}} 2^{n/4}.
\end{equation}
The trace distance can only decrease under partial trace (Lemma \ref{lem:trdist}), and more generally any completely positive trace preserving map. Hence 
\begin{equation}
    \frac12 \|\rho_A - \tilde \rho_A\|_1 \leq \frac12 \|\rho-\tilde \rho\|_1.
\end{equation}
Finally, we apply the Fannes-Audenaert inequality,
\begin{equation}
    |S_A - \tilde S_A| \leq T \log(2^{|A|}-1) +H_2(T)
\end{equation}
where $T = \frac12 \|\rho_A-\tilde \rho_A\|_1$ and $H_2(x) = -x\log x -(1-x)\log(1-x) < 1$.  It follows that 
\begin{equation}
    \begin{aligned}
        |S_A - \tilde S_A| &\leq 2\sqrt{\varepsilon_{\text{poly}}} 2^{n/4}\log(2^{|A|}) + 1 \\
        &\leq 2\sqrt{\varepsilon_{\text{poly}}} 2^{n/4} |A| + 1 \\
        &\leq 2\sqrt{\varepsilon_{\text{poly}}} 2^{n/4} n + 1 
    \end{aligned}
\end{equation}
which is $O(1)$ due to our choice of the degree $d$ of the polynomial (Eq.~\eqref{eq:degree}). Since $\tilde S_A = O(\log n)$, it follows that 
\begin{equation}
    S_A = O(\log n).
\end{equation}
Moreover, with $|A| = O(n)$, we have $S_A = O(\log |A|)$, which rules out the possibility of volume law entanglement. 
\end{proof}

Various aspects of the technical assumptions on $F(t)$ in the preceding theorem can be relaxed. In fact, the only essential property used was good approximation of $F(t)$ by polynomials. We state a reformulation of the theorem below with a slightly broader though potentially less useful set of technical assumptions. 

\begin{theorem}(Alternative to \ref{thm:1nqs})
    Suppose $\ket\Psi$ is a NQS on $n$ qubits with $\Psi(s)$ defined in Eq.~\eqref{eq:Psiofs}. Suppose $\sup_{\vec s}|\vec w^T \vec s+b| < \bar t_n$ for a possibly $n$-dependent constant $\bar t_n>0$. Suppose there exist $n$-independent constants $\alpha,\beta,\gamma>0$ such that
    \begin{equation}       \inf_{\mathrm{deg}\,\, p \leq d} \|F-p\|_{L^{\infty}([-\bar t_n,\bar t_n])} \leq \alpha e^{-\beta d^\gamma},
    \end{equation}
    where the infimum is over all polynomials of degree $\leq d$. Then 
    \begin{equation}
        S_A = O(\log n)
    \end{equation}
    and if $|A| = O(n)$, $S_A = O(\log |A|)$. 
    \label{thm:alternate1nqs}
\end{theorem}
\begin{proof}
    Only Step 2 of the proof of Theorem \ref{thm:1nqs} needs to be modified. Defining $\varepsilon_{\mathrm{poly}} = \alpha e^{-\beta d^\gamma}$, in order to achieve $2\sqrt{\varepsilon_{\mathrm{poly}}}2^{n/4}n\leq 1$, it suffices to choose choose $d = \Theta(n^{1/\gamma})$. Once again $S_A \leq \log \frac12 (d+1)(d+2) + O(1) = O(\log n)$. 
\end{proof}
\noindent The Dicke state of Theorem~\ref{thm:hh} is best viewed as a complementary explicit example of logarithmic entanglement scaling. The analyticity and bounded-domain assumptions used here are convenient sufficient conditions for obtaining a low-degree polynomial representation on a continuum domain, but they are not necessary for special examples such as the Dicke state. Indeed, in that case the affine feature $t(\vec s)=\sum_{i=1}^n s_i$ takes values only in the discrete set $\{-n,-n+2,\ldots,n\}$ (for even $n$). On this discrete set, the amplitudes can be represented exactly by a polynomial in a single affine feature. For example,
\begin{equation}
    p_n(t) =\prod_{m=1}^{n/2}\left(1-\frac{t^2}{(2m)^2}\right)
\end{equation}
satisfies $p_n(0)=1$ and $p_n(t)=0$ for every nonzero admissible value
$t\in\{-n,-n+2,\ldots,n\}$. Hence $\delta_{\sum_i s_i,0}=p_n\!\left(\sum_{i=1}^n s_i\right)$ and the Dicke state amplitudes are exactly a degree-$n$ polynomial of one affine feature. Hence the preceding arguments immediately yield $S_A=O(\log n)$ for this state, in agreement with Theorem~\ref{thm:hh}$.$

\section{Feature reduction lemma}\label{sec:frl}

The first step to go from NQS with a single nonlinearity to a general number of nonlinearities is to bound the complexity, in a suitable sense, of such a state. The purpose of this section is to establish the claim that a NQS whose network consists of $k$ nonlinearities (meaning $k$ elementary scalar nonlinear operations) must necessarily be a function of $\mu \leq k+1$ \textit{affine features}
\begin{equation}
    t_i = \vec w^T_i \vec s + b_i, \qquad (i=1,\ldots, \mu),
\end{equation}
with one technical assumption: that the graph that computes $\Psi$ is a directed acyclic graph, as we will discuss. This is most useful when $k\ll n$. In that case, though na\"ively $\Psi(\vec s) = \Psi(s_1,\ldots, s_n)$ is a function of $n$ variables, we must have that $\Psi(\vec s) = \mathcal{G}(t_1,\ldots, t_\mu)$ for some number $\mu\ll n$ of affine features. We call this \textit{feature reduction}, and it is the key to bounding entanglement in the general case. \par

We can argue for this informally as follows. Consider the first nonlinearities that the bit-string $\vec s$ ``encounters" along the computational graph of $\Psi$. Label them $\phi_1,\ldots, \phi_k$, each of which is a scalar nonlinear function. Because no nonlinearity occurred before, the inputs are $\mu$ pre-activations $t_i = \vec w_i^T \vec s + b_i$, $i=1,\ldots, \mu$, where $\mu\leq k$ is, and possibly functions of the other $\phi_i$. The rest of the computational graph computes $\Psi(\vec s)$, but only the pre-activations $t_i$ entered, with possibly one more overall affine feature. We prove this more formally with the notion of directed acyclic graphs in what follows. \par 

\begin{definition}
A directed acyclic graph (DAG) is a directed graph with no directed cycles. To each neural quantum state $\Psi(\vec s)$ whose computational graph is directed and acyclic, we refer to its DAG as DAG($\Psi$), with:
    \begin{enumerate}
        \item \textbf{Inputs: } $n$ source nodes outputting the scalars $s_1,\ldots, s_n$
        \item \textbf{Linear nodes: } each of which outputs an affine combination of its input scalars, $L(x^{(1)},\ldots, x^{(r)}) = a_0 + \sum_{j=1}^ra_j x^{(j)}$. 
        Concatenations, splits, and residual/skip connections are subsumed by these.
        \item \textbf{Nonlinear nodes (neurons): } each of which applies a scalar nonlinearity to an affine combination of its input scalars, $\phi(x^{(1)},\ldots, x^{(r)}) = \sigma(b_0 + \sum_{j=1}^r w_j x^{(j)})$. The $\sigma$'s need not be the same.
        \item \textbf{Output: } a sink node that outputs $\Psi(\vec s)$ and is affine in its inputs. 
    \end{enumerate}
\end{definition}

We refer to Fig.~\ref{fig:dag} for a simple example of a DAG$(\Psi)$. We note that any DAG has a \textit{topological sort}, which is a linear ordering of its nodes such that for any directed edge $(\phi_1,\phi_2)$ from node $\phi_1$ to node $\phi_2$, $\phi_1$ comes before $\phi_2$ in the ordering. Below, we will refer to an affine functional of the form $\ell(\vec s) = \vec u^T\vec s + c$ as an \textit{affine feature}.

\begin{lemma} (Feature reduction) Let $\ket\Psi = \sum_{\vec s}\Psi(\vec s)\ket{\vec s}$ be a NQS acting on $n$ spins whose computational graph is a directed acyclic graph with $k$ scalar nonlinear nodes. Then there exist affine features $t_1,\ldots, t_\mu$ together with a function $\mathcal{G}:\mathbb{R}^{\mu} \to \mathbb{C}$ such that
\begin{equation}
    \Psi(\vec s) = G(t_1(\vec s), \ldots, t_\mu(\vec s))
\end{equation}
with $\mu\leq k+1$. In particular, $\Psi$ depends on at most $k+1$ affine features of $\vec s$. 
\label{lem:featurereduction}
\end{lemma}
\begin{proof}
    Let $T = (v_1, \ldots, v_w)$ be a topological sort of DAG($\Psi$). Let $N = (\phi_1,\ldots, \phi_k)$ be the sorted nonlinear nodes in $T$. For each $\phi_j$, its pre-activation can be written as
    \begin{equation}
        a_j(\vec s) = \vec u_j^T \vec s + c_j + \Lambda_j(\phi_1(\vec s),\ldots, \phi_{j-1}(\vec s))
    \end{equation}
    where $\Lambda_j$ is affine in its arguments. Define the \textit{new affine feature} contributing to $\phi_j$ to be 
    \begin{equation}
        t_j(\vec s) =  \vec u_j^T \vec s + c_j.
    \end{equation}
    We claim that, along with some affine feature $t_0(\vec s) = \vec u_0^T\vec s + c_0$, there is a subcollection $\{t_{j_0},\ldots, t_{j_\mu}\}$ with $\mu \leq k+1$ such that     \begin{equation}\label{eq:PsiAffine}
        \Psi(\vec s) = \mathcal{G}(t_{j_0}(\vec s),\ldots, t_{j_\mu}(\vec s))
    \end{equation}
    for some function $\mathcal{G}:\mathbb{R}^{\mu}\to \mathbb{C}$. We prove this by induction. Clearly $\phi_1$ has a pre-activation 
    \begin{equation}
        a_1(\vec s) = \vec u_1^T\vec s + c_1 = t_1(\vec s)
    \end{equation}
    so $\phi_1 = \phi_1(t_1(\vec s))$. Next, suppose $\phi_i(\vec s) = F_i(t_1(\vec s), \ldots, t_i(\vec s))$ for all $i<j$. The pre-activation for $\phi_j$ can be written as 
    \begin{equation}
        a_j(\vec s) = t_j(\vec s) +\Lambda_j(\phi_1(\vec s),\ldots, \phi_{j-1}(\vec s)) = t_j(\vec s) +\tilde\Lambda_j(t_1(\vec s),\ldots, t_{j-1}(\vec s))
    \end{equation}
    so $\phi_j(\vec s) = F_j(t_1(\vec s),\ldots, t_j(\vec s))$. Thus each new nonlinearity adds at most one new affine feature $t_j$. The output node is affine in its inputs, and hence equals 
    \begin{equation}
        \Psi(\vec s) = t_0(\vec s) + \Gamma(\phi_1(\vec s) ,\ldots, \phi_k(\vec s)) = t_0(\vec s) + \tilde \Gamma(t_1(\vec s),\ldots, t_k(\vec s))
    \end{equation}
    for some affine feature $t_0$. Finally, the $t_0,\ldots, t_k$ may not be linearly independent. Hence there is a subcollection $\{t_{j_0},\ldots, t_{j_\mu}\}$ with $\mu\leq k+1$ such that Eq.~\eqref{eq:PsiAffine} holds. 
\end{proof}

\begin{remark}
Requiring $\Psi$ to be a DAG is not a severe restriction. For instance, MLPs, CNNs, transformers, and all feed-forward neural networks are DAGs. Recurrent neural networks are not included in general, unless unrolled over $T$ time steps,
in which case they become a feed-forward DAG with an effective number of scalar nonlinearities $k_{\rm eff}$ scaling with $T$.
\end{remark}

Establishing that any NQS $\Psi(\vec s)$ whose neural network is a directed acyclic graph with $O(1)$ nonlinearities is a function only $O(1)$ affine features $\vec u^T \vec s + c$ of the input greatly reduces the difficulty of studying general bounded-nonlinearity NQS, which we turn to next. 

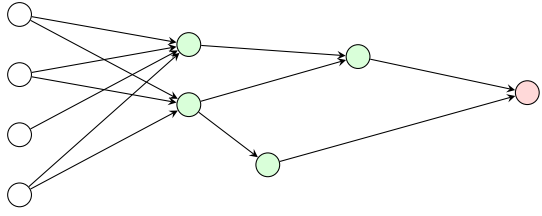
\begin{figure}
\begin{tikzpicture}[x=1.5cm, y=0.8cm, >=stealth]
  \tikzstyle{node}=[circle,draw=black,minimum size=9pt,inner sep=0pt]
  \tikzstyle{input}=[node,fill=white!15]
  \tikzstyle{hidden}=[node,fill=green!15]
  \tikzstyle{output}=[node,fill=red!15]

  \node[input] (I1) at (0,-0.5) {};
  \node[input] (I2) at (0,-1.5) {};
  \node[input] (I3) at (0,-2.5) {};
  \node[input] (I4) at (0,-3.5) {};

  \node[hidden] (H1) at (1.5,-1.0) {};
  \node[hidden] (H2) at (1.5,-2.0) {};
  \node[hidden] (H3) at (3.0,-1.2) {};
  \node[hidden] (H4) at (2.2,-3.0) {};

  \node[output] (O) at (4.5,-1.8) {};

  \draw[->,thin] (I1) -- (H1);
  \draw[->,thin] (I2) -- (H1);
  \draw[->,thin] (I3) -- (H1);
  \draw[->,thin] (I4) -- (H1);

  \draw[->,thin] (I1) -- (H2);
  \draw[->,thin] (I2) -- (H2);
  \draw[->,thin] (I4) -- (H2);

  \draw[->,thin] (H1) -- (H3);
  \draw[->,thin] (H2) -- (H3);
  \draw[->,thin] (H2) -- (H4);

  \draw[->,thin] (H3) -- (O);
  \draw[->,thin] (H4) -- (O);

\end{tikzpicture}
\caption{Example of a directed acyclic graph (DAG) representation of a feed-forward neural quantum state. Inputs (white) are affinely combined and passed through intermediate nonlinear nodes (green), ultimately producing the output (red).}
\label{fig:dag}
\end{figure}

\section{Proof of entanglement bound for NQS with $k$ nonlinearities}
\label{sec:proofO1}

In this section, we generalize Theorem \ref{thm:1nqs}, which was for a single nonlinearity, to a feed-forward neural quantum state with $k$ nonlinearities. We'll be most interested in the case when $k$ is $O(1)$, but we will keep $k$ general. Once again, 
\begin{equation}
    \ket{\Psi} = \sum_{\vec s} \Psi(\vec s)\ket{\vec s}
\end{equation}
but we do \textit{not} assume a particular form for $\Psi(\vec s)$ other than that it corresponds to a feed-forward neural network with $k$ nonlinearities whose computational graph is directed and acyclic, and certain constraints imposed by our technical analyticity assumptions. By the feature reduction lemma (Lemma \ref{lem:featurereduction}), 
\begin{equation}\label{eq:PsitoG}
    \Psi(\vec s) = \mathcal{G}(t_1(\vec s),\ldots, t_\mu(\vec s))
\end{equation}
for $\mu\leq k+1$ where each pre-activation is of the form $t_i(\vec s) = \vec w_i^T \vec s + b_i$. 

\begin{theorem}\label{thm:O1nqs}
    Suppose $\ket{\Psi} = \sum_{\vec s} \Psi(\vec s) \ket{\vec s}$ is the feed-forward NQS on $n$ qubits, so Eq.~\eqref{eq:PsitoG} holds. Suppose $\sup_{\vec s,i}|\vec w_i^T \vec s + b_i|<\bar t$ for an $n$-independent constant $\bar t>0$. Suppose the function $g(x_1,\ldots, x_\mu) = \mathcal{G}(\bar tx_1, \ldots, \bar t x_\mu)$ admits an analytic extension $g(z_1,\ldots, z_\mu)$ to a product of Berstein ellipses $\mathcal{B} = B(a_1)\times \cdots \times B(a_\mu)$, where $a_i>0$, and satisfies $|g(z_1,\ldots, z_\mu)|\leq C$ for some $O(1)$ constant $C$ on $\mathcal{B}$. Denote by $S_A$ the von Neumann entropy of the reduced density matrix on a subset $A\subset [n]$. Then 
    \begin{equation}
        S_A = O(k\log n).
    \end{equation}
In particular, no such NQS with $k\sim O(1)$ nonlinearities can achieve volume law entanglement. 
\end{theorem}
\begin{proof}
We proceed with the same 3 steps as in Theorem \ref{thm:1nqs}, with modifications where needed.\par 
\textbf{Step 1.} Let $A \subset [n]$ and let $\bar A$ be its complement. Let us refer to 
\begin{equation}
    z_{i,s} = t_i = \vec w_i^T\vec s+b_i
\end{equation}
interchangeably. Let $\mathcal{H}_A, \mathcal{H}_{\bar A}$ denote the Hilbert spaces on $A$ and $\bar A$ and let $\ket{\vec u}, \ket{\vec v}$ be $Z$-basis states on $A$ and $\bar A$, respectively. Let us define \textit{split affine features} $x_{i,u},y_{i,v} \in\mathbb{R}$ by $x_{i,u}= \sum_{j\in A} w_{i,j} u_j +b_i/2$, $y_{i,v} = \sum_{j\in \bar A} w_{i,j} v_j + b_i/2$, which satisfy $x_{i, u}+y_{i, v} = z_{i,s}$. Using Eq.~\eqref{eq:PsitoG}, we may write 
\begin{equation}
    \ket\Psi = \sum_{\vec u,\vec v} M_{\vec u,\vec v}\ket{\vec u}_A\ket{\vec v}_{\bar A},\qquad M_{\vec u,\vec v} =  \mathcal{G}(x_{1,u} + y_{1,v},\ldots, x_{\mu,u}+y_{\mu,v}).
\end{equation}
Then
\begin{equation}
\begin{aligned}
    \rho_A &= \Tr_{\bar A}\ket{\Psi}\bra{\Psi} = \sum_{\vec u,\vec u',\vec v} M_{\vec u,\vec v}M^*_{\vec u',\vec v} \ket{\vec u}_A\bra{\vec u'}_A\\
    &= \sum_{\vec u,\vec u'} (MM^\dagger)_{\vec u,\vec u'}\ket{\vec u}_A\bra{\vec u'}_A. 
\end{aligned}
\end{equation}
and the von Neumann entanglement entropy once again satisfies $S_A \leq \log \rank M$. \par 

\textbf{Step 2. } We have assumed the pre-activations are bounded by some constant $\bar t>0$. Moreover, we have assumed that $g(x_1,\ldots, x_\mu) = \mathcal{G}(\bar tx_1,\ldots, \bar tx_\mu)$ can be analytically extended from $[-1,1]^\mu$ to the product of Bernstein ellipses $\mathcal{B} = B(a_1)\times \cdots \times B(a_\mu)$ for $a_i >0$, where it satisfies a bound $|g|<C$. Then there exists a multivariable polynomial $S_d(\vec x)$ with degree no more than $d$ in each variable such that for all $\vec x \in [-1,1]^\mu$,
\begin{equation}
    \left|g(\vec x) - S_d(\vec x)\right| \leq C\mu \left(\frac{2r_*}{r_*-1}\right)^\mu  r_*^{-d-1}
\end{equation}
where $r_* = \min_j e^{a_j}$, by Lemma \ref{lemma:multivariable}. It follows that 
\begin{equation}
    |\mathcal{G}(\vec t) - P_d(\vec t)| \leq  C\mu \left(\frac{2r_*}{r_*-1}\right)^\mu  r_*^{-d-1} \equiv \varepsilon_{\mathrm{poly}}
\end{equation}
for $\vec t = (t_1,\ldots, t_\mu) \in [-\bar t,\bar t]^\mu$ for a multivariable polynomial $P_d$ with degree no more than $d$ in each variable. At a later stage, we will need to bound $2\sqrt{\varepsilon_{\mathrm{poly}}} 2^{n/4} n$ by an $O(1)$ constant. With this in mind, it suffices to take $d$ to satisfy 
\begin{equation}\label{eq:degree2}
    d \geq \frac{n}{2\log r_*}\log 2 + \frac{2\log n}{\log r_*} + \frac{\log(C2^{\mu+2}\mu(r_*-1)^{-\mu}r_*^{\mu-1})}{\log r_*} = \Theta(n)
\end{equation}
which implies $2\sqrt{\varepsilon_{\mathrm{poly}}} 2^{n/4} n \leq 1$.

\textbf{Step 3. } In what follows, write 
\begin{equation}
    \mathcal{G}(z_{1,s},\ldots, z_{\mu,s}) = P_d(z_{1,s},\ldots, z_{\mu,s}) + \Delta_s
\end{equation}
and define $Q = \|P_d\|_2^2 \equiv \sum_{\vec s} |P_d(z_{1,s},\ldots, z_{\mu,s})|^2$, $\|\Delta\|_2^2 = \sum_{\vec s}|\Delta_s|^2\leq 2^n\varepsilon_{\mathrm{poly}}^2$ and $\|\mathcal{G}\|_2^2=\sum_{\vec s}|\mathcal{G}(z_{1,s},\ldots, z_{\mu,s})|^2=1$. We define the normalized auxiliary state
\begin{equation}
    \ket{\tilde\Psi} = Q^{-1/2} \sum_{\vec s}P_d(z_{1,s},\ldots, z_{\mu,s}) \ket{\vec s} = Q^{-1/2} \sum_{\vec u,\vec v} \widetilde M_{\vec u,\vec v}\ket{\vec u}_A\ket{\vec v}_{\bar A},
\end{equation}
where $\widetilde M_{\vec u,\vec v} = P_d(x_{1,u}+y_{1,v},\ldots, x_{\mu,u}+y_{\mu,v})$. Let $\tilde S_A$ be the von Neumann entanglement entropy of $\tilde \rho_A = \Tr_{\bar A} \ket{\tilde \Psi}\bra{\tilde \Psi}$, which again satisfies $\tilde S_A \leq \log\rank \widetilde M$. We can bound the rank of $\widetilde M$ as follows. We will write 
\begin{equation}
    P_d(\vec t) = \sum_{0\leq n_1,\ldots, n_\mu \leq d} \alpha_{\vec n} t_1^{n_1} \cdots t_\mu^{n_\mu}
\end{equation}
for the expansion of $P_d$, where $\vec n = (n_1,\ldots, n_\mu)$. We have 
\begin{subequations}
    \begin{align}
        \widetilde M &= \sum_{\vec u,\vec v} \widetilde M_{\vec u,\vec v} \ket{\vec u}\bra{\vec v}\\
        &= \sum_{\vec u,\vec v} \sum_{0\leq n_1,\ldots, n_\mu \leq d} \alpha_{\vec n} (x_{1,u}+y_{1,v})^{n_1} \cdots (x_{\mu,u}+y_{\mu,v})^{n_\mu} \ket{\vec u}\bra{\vec v}\\
        &= \sum_{\vec u,\vec v} \sum_{0\leq n_1,\ldots, n_\mu \leq d} \alpha_{\vec n} \prod_{\ell=1}^\mu \sum_{k_\ell=0}^{n_\ell} \binom{n_\ell}{k_\ell} x_{\ell,u}^{n_\ell-k_\ell}y_{\ell,v}^{k_\ell} \ket{\vec u}\bra{\vec v}\\
        &= \sum_{0\leq n_1,\ldots, n_\mu \leq d}\sum_{k_1,\ldots, k_\mu=0}^{n_1,\ldots,n_\mu} \alpha_{\vec n} \left(\sum_{\vec u}\prod_{\ell=1}^\mu \binom{n_\ell}{k_\ell} x_{\ell,u}^{n_\ell-k_\ell} \ket{\vec u}\right)\left(\sum_{\vec v}\prod_{\ell=1}^\mu y^{k_\ell}_{\ell,v}\bra{\vec v}\right)\\
        &\equiv \sum_{0\leq n_1,\ldots, n_\mu \leq d}\sum_{k_1,\ldots, k_\mu=0}^{n_1,\ldots,n_\mu}\tilde m_{\vec n, k_1,\ldots, k_\mu}^{\text{rank-1}}
    \end{align}
\end{subequations}
where we've expanded using the binomial theorem and reorganized as a sum over rank-1 matrices $\tilde m_{\vec n, k_1,\ldots, k_\mu}^{\text{rank-1}}$. By subadditivity of rank, 
\begin{equation}
    \rank \widetilde M \leq  \sum_{0\leq n_1,\ldots, n_\mu \leq d}\sum_{k_1,\ldots, k_\mu=0}^{n_1,\ldots,n_\mu} 1 = \left((d+1)(d+2)/2\right)^\mu
\end{equation}
which implies $\tilde S_A \leq \log \left((d+1)(d+2)/2\right)^\mu = O(\mu\log n)$, since $d$ is $\Theta(n)$. The rest of the proof proceeds just as in Theorem \ref{thm:1nqs}, but we reproduce it for completeness. We define $\ket{\delta} = \ket{\Psi} - \ket{\tilde \Psi}$ and wish to bound $\|\ket{\delta}\|_2^2$. We have
\begin{subequations}
    \begin{align}        \|\ket\delta\|_2^2 &= 2-2\Re \langle \Psi|\tilde \Psi\rangle \\
    &\leq 2 |1-Q^{1/2}| + 2\left| Q^{-1/2} \sum_s \Delta_s^* P_d(z_s)\right|.
    \end{align}
\end{subequations}
We bound the two terms individually. First, using $\|\mathcal{G}\|_2 \leq \|P_d\|_2 + \|\Delta\|_2$ we have $|1-Q^{1/2}| = | \|\mathcal{G}\|_2-\|P_d\|_2| \leq \|\Delta\|_2$. Second, using the Cauchy-Schwarz inequality we have 
\begin{equation}
    \left| Q^{-1/2} \sum_s \Delta_s^* P_d(z_s)\right| \leq Q^{-1/2} \|\Delta\|_2 \|P_d\|_2 = \|\Delta \|_2. 
\end{equation}
We conclude 
\begin{subequations}
    \begin{align}        \|\ket\delta\|_2^2 &\leq 2\cdot 2 \cdot \|\Delta\|_2 \\
    &\leq 4\varepsilon_{\mathrm{poly}}2^{n/2}\\
    \to\|\ket\delta\|_2 &\leq 2\sqrt{\varepsilon_{\mathrm{poly}}}2^{n/4}.
    \end{align}
\end{subequations}
Next, we wish to bound the trace distance $\frac12 \|\rho_A-\tilde \rho_A\|_1$. From Lemma \ref{lemma:fvdg} we have 
\begin{equation}
    \frac12 \|\rho-\tilde \rho\|_1 \leq \|\ket\delta\|_2 \leq 2\sqrt{\varepsilon_{\text{poly}}} 2^{n/4}.
\end{equation}
The trace distance can only decrease under partial trace (Lemma \ref{lem:trdist}), and more generally any completely positive trace preserving map. Hence 
\begin{equation}
    \frac12 \|\rho_A - \tilde \rho_A\|_1 \leq \frac12 \|\rho-\tilde \rho\|_1.
\end{equation}
Finally, we apply the Fannes-Audenaert inequality,
\begin{equation}
    |S_A - \tilde S_A| \leq T \log(2^{|A|}-1) +H_2(T)
\end{equation}
where $T = \frac12 \|\rho_A-\tilde \rho_A\|_1$ and $H_2(x) = -x\log x -(1-x)\log(1-x) < 1$.  It follows that 
\begin{equation}
    \begin{aligned}
        |S_A - \tilde S_A| &\leq 2\sqrt{\varepsilon_{\text{poly}}} 2^{n/4}\log(2^{|A|}) + 1 \\
        &\leq 2\sqrt{\varepsilon_{\text{poly}}} 2^{n/4} |A| + 1 \\
        &\leq 2\sqrt{\varepsilon_{\text{poly}}} 2^{n/4} n + 1 
    \end{aligned}
\end{equation}
which is $O(1)$ due to our choice of the $d$. Since $\tilde S_A = O(\mu \log n)$, it follows that $S_A = O(\mu \log n)$. Since $\mu \leq k+1$, we may write this as
\begin{equation}
    S_A = O(k \log n).
\end{equation}
Moreover, with $|A| = O(n)$, we have $S_A = O(k\log |A|)$, which rules out the possibility of volume law entanglement if $k\sim O(1)$. 
\end{proof}

As for the single nonlinearity case, various aspects of the technical assumptions on $G(\vec t)$ in the preceding theorem can be relaxed. In fact, the only essential property used was good approximation of $G(\vec t)$ by polynomials. We state a reformulation of the theorem below with a slightly broader, though potentially less useful  set of technical assumptions. 
\begin{theorem}(Alternative to \ref{thm:O1nqs})
    Suppose $\ket\Psi$ is a NQS on $n$ qubits with $\Psi(s)$ defined as in the setting of Theorem \ref{thm:O1nqs}. Suppose $\sup_{\vec s,i}|\vec w_i^T \vec s+b| < \bar t_n$ for a possibly $n$-dependent constant $\bar t_n>0$. Suppose there exist $n$-independent constants $\alpha,\beta,\gamma>0$ such that
    \begin{equation}       \inf_{\mathrm{deg}\,\, p \leq d} \|G-p\|_{L^{\infty}([-\bar t_n,\bar t_n]^\mu)} \leq \alpha e^{-\beta d^\gamma},
    \end{equation}
    where the infimum is over all polynomials of degree $\leq d$ in each variable. Then 
    \begin{equation}
        S_A = O(k\log n).
    \end{equation}
\end{theorem}
\begin{proof}
    Identical, up to minor modifications, to the proof of Theorem \ref{thm:alternate1nqs}. 
\end{proof}

\begin{corollary}(Finite-horizon recurrent NQS)
If a recurrent NQS is evaluated for $T$ steps and can be unrolled into a feed-forward DAG with $k_{\rm eff}$ scalar nonlinear nodes,
then under the same bounded-preactivation and analyticity/boundedness assumptions (uniformly over the unrolled graph),
the entanglement bound holds with $k$ replaced by $k_{\rm eff}$:
\begin{equation}
S_A = O(k_{\rm eff}\log n).
\end{equation}
In particular, for a recurrent cell with $k_{\rm cell}$ scalar nonlinearities per step, $k_{\rm eff}\sim T\,k_{\rm cell}$.
\end{corollary}

\begin{remark}
Theorems \ref{thm:1nqs} and \ref{thm:O1nqs} as stated do not apply to NQS built from ReLUs or other nonanalytic nonlinearities. However, as seen in our numerical experiments, SN-NQS with ReLUs exhibit the same entanglement scaling as the analytic nonlinearities. This suggests that the results here can be extended to piecewise analytic NQS (e.g. ReLU-networks), but we do not pursue this further here. In any case, $\text{ReLU}(x)$ can be well-approximated by analytic nonlinearities such as $\text{softplus}_\beta(x) = (1/\beta) \log(1+e^{\beta x})$.
\end{remark}

\section{Neural quantum state variational Monte Carlo}\label{sec:nqsvmc}

In the main text, we claimed that a family of NQS on $n$ qubits but with a bounded $k\sim O(1)$ number are even more efficient than general NQS, polynomially reducing the computational scaling of variational Monte Carlo (VMC). In this section, we explain this claim. \par 

Neural quantum states, coupled with VMC, offer a powerful variational technique for finding ground states of quantum systems. Given a parameterized NQS $\Psi_\theta$, the ground state of a Hamiltonian $H$ can be approximated by minimizing $\langle \Psi_\theta|H|\Psi_\theta\rangle$, the expectation value of the energy, with respect to $\theta$, the set of free parameters. In the case of NQS, $\theta$ represents the network parameters and the optimization is typically gradient-based using a differentiable representation of the energy expectation value. If basis states are labelled by $\vec s$, we may write
\begin{equation}
    \min_\theta \langle \Psi_\theta|H|\Psi_\theta\rangle = \min_\theta \mathbb{E}_{\vec s\sim|\Psi_\theta(\vec s)|^2}\left( \sum_{\vec s'} \langle \vec s | H|\vec s'\rangle \frac{\Psi_\theta(\vec s')}{\Psi_\theta(\vec s)}\right)
\end{equation}
and define the local energy $E_{\mathrm{loc}}(\vec s) \equiv \sum_{\vec s'} \langle \vec s | H|\vec s'\rangle \frac{\Psi_\theta(\vec s')}{\Psi_\theta(\vec s)}$. The most taxing step of NQS-VMC is computing local energy. Assuming a local or sufficiently sparse Hamiltonian (\textit{i.e.} $O(n)$ entries per row), this requires $O(n)$ forward passes through $\Psi_\theta$. Finally, we compute the gradient as 
\begin{equation}
    \pdv{E_\theta}{\theta} \approx \sum_{i=1}^K (E_{\text{loc}}(\vec s^{(i)}) - E_\theta)^*\pdv{\log \Psi_\theta(\vec s^{(i)})}{\theta}
\end{equation}
where $\vec s^{(i)}$ are the $K$ samples used to evaluate $\mathbb{E}_{s\sim |\Psi_\theta|^2}$. The cost of sampling $s\sim |\Psi_\theta|^2$ and the cost of $\pdv{\log \Psi_\theta(s^{(i)})}{\theta}$ are the cost of a single forward pass each. Since each forward pass processes an input of length $n$, the cost scales roughly as at least $\Omega(n|\text{params}|)$ for most commonly used architectures~\cite{sharir2022towards}. This implies an overall cost of $\Omega(n^2|\text{params}|)$ for a typical NQS, driven by the linear number of forward passes for
computing the local energy. We refer the reader to Ref.~\cite{sharir2022towards} for further description of the scaling. \par 

Now consider a family of feed-forward NQS on $n$ qubits with a bounded $k\sim O(1)$ number of nonlinearities. As established by Lemma \ref{lem:featurereduction}, such an NQS can be expressed as a function of $\mu\leq k+1$ affine features: 
\begin{equation}
    \Psi_\theta(\vec s) = \mathcal{G}_\theta(t_1(\vec s),\ldots, t_\mu(\vec s)). 
\end{equation}
Therefore, each forward pass need only process an input of length $\mu \sim O(1)$ and its cost scales roughly as at least $\Omega(|\text{params}|)$. Therefore, the overall cost in this case is $\Omega(n|\text{params}|)$, a linear-in-$n$ improvement from the typical case scaling, without even counting the possible reduction in $|\text{params}|$.

\section{Other lemmas}
\label{sec:lems}
\begin{lemma}\label{lemma:bernstein}(Theorem 8.2 in \cite{trefethen}) Let $f$ be a function on $[-1,1]$ such that it can be extended analytically to the \textit{Bernstein ellipse}
\begin{equation}
    B(a) = \left\{z=x+iy \in \mathbb{C} : \frac{x^2}{\cosh^2(a)} + \frac{y^2}{\sinh^2(a)} < 1\right\}
\end{equation}
where $a>0$ and where $|f(z)|<M$ for $z\in B(a)$ and some $M$. Then with $\rho = e^{a} >1$, for $x\in [-1,1]$, we have
\begin{equation}
    \left|f(x) - \sum_{n=0}^{d}a_nT_n(x) \right| \leq \frac{2M \rho^{-d}}{\rho-1}
\end{equation}
where $T_n(x)$ is the $n'$th Chebyshev polynomial and the coefficients satisfy $|a_n|\leq 2M \rho^{-n}$. 
\end{lemma}

\begin{lemma}\label{lemma:multivariable} (Multivariable extension of Lemma \ref{lemma:bernstein}) Let $f$ be a function on $[-1,1]^k$ such that it can be extended analytically to the product of Bernstein ellipses 
\begin{equation}
    \mathcal{B} = B(a_1) \times \cdots \times B(a_k)
\end{equation}
where $a_j>0$ and assume $\sup_{ \mathcal{B}} |f|<M$. Let $\rho_j = e^{a_j}>1$ and let $\rho_* = \min_j \rho_j$. Then for $\vec x \in [-1,1]^k$ we have 
\begin{equation}
    \left|f(\vec x) - S_d(\vec x)\right| \leq C_{M,k,\rho_*} \, \rho_*^{-d}
\end{equation}
where $C_{M,k,\rho_*} = Mk\rho_{*}^{-1}\left( \frac{2\rho_*}{\rho_*-1} \right)^k$ and
\begin{equation} \label{eq:SDtruncation}
S_d(\vec x) = \sum_{0\leq n_j\leq d}c_{\vec n}\prod_{j=1}^kT_{n_j}(x_j)
\end{equation}
is a multivariable polynomial of degree no more than $d$ in each variable, with $T_n(x)$ the $n$th Chebyshev polynomial.
\end{lemma}

\begin{proof}
First, fix $(x_2,\ldots, x_k)$. By Lemma \ref{lemma:bernstein}, $f(x_1,x_2,\ldots,x_k) = \sum_{n_1=0}^{\infty} c_{n_1}(x_2,\ldots, x_k) T_{n_1}(x_1)$ where $|c_{n_1}|\leq 2M\rho_1^{-n_1}$. Moreover,
\begin{equation}
    \left| f(x_1,\ldots, x_k) - \sum_{n_1=0}^{d} c_{n_1}(x_2,\ldots, x_k) T_{n_1}(x_1)\right|\leq \frac{2M\rho_1^{-d}}{\rho_1-1}
\end{equation}
uniformly in $x_1 \in [-1,1]$, where the coefficients $c_{n_1}(x_2,\ldots, x_k)$ extend analytically to $B(a_2)\times \cdots \times B(a_k)$~\cite{trefethen}. Next, fix $n_1$ and $(x_3,\ldots, x_k)$. Applying Lemma \ref{lemma:bernstein}, $c_{n_1}(x_2,\ldots, x_k) = \sum_{n_2} c_{n_1,n_2}(x_3,\ldots, x_k) T_{n_2}(x_2)$ with bounds $|c_{n_1,n_2}| \leq 2(2M\rho_{1}^{-n_1}) \rho_2^{-n_2} =2^2M\rho_1^{-n_1}\rho_2^{-n_2}$. Moreover, the $c_{n_1,n_2}(x_3,\ldots, x_k)$ extend analytically to $B(a_3)\times \cdots \times B(a_k)$. \par 
Proceeding inductively, we have
\begin{equation}
    f(\vec x) = \sum_{\vec n = (n_1,\ldots, n_k)} c_{\vec n} \prod_{j=1}^k T_{n_j}(x_j)
\end{equation}
where $|c_{\vec n}| < 2^k M\prod_{j=1}^k \rho_j^{-n_j}$. Defining the truncation $S_d(\vec x)$ as in Eq.~\eqref{eq:SDtruncation} and using $|T_n(x)|<1$ on the interval, we have 
\begin{subequations}
    \begin{align}
        \left|f(\vec x)-S_d(\vec x)\right| &\leq \sum_{\vec n \mid \exists n_j > d}|c_{\vec n}|  \\
        &\leq \sum_{\vec n \mid \exists n_j > d} 2^k M \prod_{j=1}^k \rho_j^{-n_j}\\
        &\leq \sum_{\ell=1}^k \,\,\sum_{\vec n\mid n_\ell > d} 2^k M \prod_{j=1}^k \rho_j^{-n_j}\\
        &\leq \sum_{\ell=1}^k 2^k M \left(\prod_{j\neq \ell} \frac{\rho_j}{\rho_j-1}\right) \frac{\rho_\ell^{-d}}{\rho_\ell-1}
    \end{align}
\end{subequations}
uniformly on $\vec x \in [-1,1]^k$. The first sum is over vectors $\vec n$ with at least one coordinate $>d$. In the third step, we bounded this sum by a union bound over the $k$ regions $n_j>d$. Finally, letting $\rho_* = \min_j \rho_j$ and using the fact that $\frac{\rho}{\rho-1}$ and $\rho^{-1}$ are both monotonically decreasing, we have 
\begin{subequations}
    \begin{align}
        \left|f(\vec x)-S_d(\vec x)\right| &\leq \sum_{\ell=1}^k 2^k M \left(\prod_{j\neq \ell} \frac{\rho_*}{\rho_*-1}\right) \frac{\rho_*^{-d}}{\rho_*-1}\\
        & = 2^k k M\left(\frac{\rho_*}{\rho_*-1}\right)^k \rho_*^{-d-1}
    \end{align}
\end{subequations}
as desired.

\end{proof}

\begin{lemma}\label{lemma:fvdg}
For pure state density matrices $\rho = \ketbra{a}{a}, \sigma = \ketbra{b}{b}$, we have $\frac12\|\rho-\sigma\|_1 \leq \|\ket{a}-\ket{b}\|_2$. 
\end{lemma}
\begin{proof}
Let $\alpha = \braket{a}{b}$. We have $\frac12\|\rho-\sigma\|_1 = \frac12 \tr{\sqrt{(\rho-\sigma)^\dagger(\rho-\sigma)}} = \sqrt{1-|\alpha|^2}$. Since $|\alpha|<1$, we have $1-|\alpha|^2  = (1+|\alpha|)(1-|\alpha|) \leq 2(1-|\alpha|) \leq 2(1-\Re \alpha) = \|\ket a- \ket b\|_2^2$. Hence $\frac12\|\rho-\sigma\|_1\leq \|\ket a- \ket b\|_2$. 
\end{proof}

\begin{lemma}\label{lem:trdist}
    For two density matrices $\rho_{AB}$ and $\sigma_{AB}$ and their reduced density matrices $\rho_A,\sigma_A$, we have 
    \begin{equation}
        \|\rho_A-\sigma_A\|_1 \leq \|\rho_{AB}-\sigma_{AB}\|_1. 
    \end{equation}
    That is, the trace distance can only decrease under partial trace. 
\end{lemma}
\begin{proof}
    The trace distance has the variational definition $\|\rho\|_1 = \sup_{-I \leq U\leq I} \Tr[U\rho]$. Then it follows that 
    \begin{subequations}
        \begin{align}
            \|\rho_A-\sigma_A\|_1 &= \sup_{-I_A\leq U_A\leq I_A} \Tr[U_A(\rho_A-\sigma_A)]\\
            &= \sup_{-I_A\leq U_A\leq I_A} \Tr[(U_A \otimes I_B)(\rho_{AB}-\sigma_{AB})]\\
            &\leq \sup_{-I_{AB}\leq U_{AB}\leq I_{AB}} \Tr[U_{AB}(\rho_{AB}-\sigma_{AB})]\\
            &= \|\rho_{AB}-\sigma_{AB}\|_1
        \end{align}
    \end{subequations}
where in the second line we've used $\Tr[\Tr_A[U_{AB}]] = \Tr[U_{AB}]$ and in the third line we've used the fact that matrices of the form $U_A\otimes I_B$ are a subset of matrices of the form $U_{AB}$. 
\end{proof}

\bibliography{ref}